\documentclass[acmsmall,screen]{acmart}
\usepackage{macros}

\extendedtrue

\title{Fuss-Free Cumulative Universes: Theory and Practice}
\subtitle{Extended Version with Supplementary Appendices}

\author{Raphaël Sterbac}
\orcid{0009-0001-6779-8323}
\affiliation{%
  \department{Computer Science Department}
  \institution{ENS Paris-Saclay -- Université Paris-Saclay}
  \city{Gif-sur-Yvette}
  \country{France}
}
\email{raphael.sterbac@ens-paris-saclay.fr}

\author{Jonathan Sterling}
\correspondingauthor
\orcid{0000-0002-0585-5564}
\affiliation{%
  \department{Computer Laboratory}
  \institution{University of Cambridge}
  \city{Cambridge}
  \country{United Kingdom}
}
\email{js2878@cl.cam.ac.uk}

\setcopyright{cc}
\setcctype{by}
\acmDOI{10.1145/3871782}
\acmYear{2027}
\acmJournal{PACMPL}
\acmVolume{11}
\acmNumber{POPL}
\acmArticle{80}
\acmMonth{1}
\acmSubmissionID{popl27main-p752-p}
\received{2026-07-09}
\received[accepted]{2026-10-05}

\ccsdesc[500]{Theory of computation~Type theory}
\ccsdesc[100]{Theory of computation~Categorical semantics}
\ccsdesc[300]{Theory of computation~Type structures}
\ccsdesc[100]{Theory of computation~Logic and verification}

\keywords{dependent type theory, universes, cumulativity, polymorphism, elaboration, normalisation}

\begin{document}

\begin{abstract}
Universes are central to dependent type theory, and they are notoriously difficult to handle in a way that is both correct and usable. We propose a new \emph{``fuss-free''} generalised algebraic presentation for polymorphic cumulative universes that dispenses with the intricate theory of coherent universe coercions in favour of a simpler formulation, which we prove equivalent by means of a normalisation theorem for the former. We study the fuss-free formulation by means of (1) an abstract specification of its bidirectional elaboration algorithm, and (2) a concrete implementation in Haskell. We also describe and implement an extension of the fuss-free universe hierarchy with a judgemental notion of datatype description from which prior notions of cumulative inductive type may be derived. \end{abstract}

\maketitle

\section{Introduction}\label{sec:introduction}

In dependently typed programming languages like Lean, Rocq, and Agda, \emph{universes} turn types into ordinary data that can be taken as arguments, returned from functions, and even embedded in more complex data structures. This is achieved by introducing a type $\TpUniv$ whose elements encode types:
\begin{mathparpagebreakable}
  \ebrule[univ-formation]{
    \infer0{\Gamma\vdash \TpUniv\IsType}
  }
  \and
  \ebrule[univ-decoding]{
    \hypo{\Gamma\vdash A:\TpUniv}
    \infer1{\Gamma\vdash \TpEl\prn{A}\IsType}
  }
\end{mathparpagebreakable}

It is important that we do not have an inverse to the decoding function $\TpEl$, for if $\TpUniv$ can be encoded by an element of $\TpUniv$, then the type system becomes inconsistent in the sense that every type becomes inhabited. More importantly, the intended semantic models of type theory (\emph{e.g.}\ sets!) would be lost. Instead, we must arrange for the universe $\TpUniv$ to be closed under the \emph{connectives} of type theory; for example, in addition to the usual formation rule for the dependent function space, we need a separate rule targeting the universe:
\begin{mathparpagebreakable}
  \ebrule[dfun-formation]{
    \hypo{\Gamma\vdash A\IsType}
 -2021   \hypo{\Gamma, x:A\vdash B\brk{x}\ \mathit{type}}
    \infer2{\Gamma\vdash \prn{x:A}\to B\brk{x}\ \mathit{type}}
  }
  \and
  \ebrule[dfun-code-formation]{
    \hypo{\Gamma\vdash M:\TpUniv}
    \hypo{\Gamma, x:\TpEl\prn{M}\vdash N\brk{x} : \TpUniv}
    \infer2{\Gamma\vdash\prn{x:M}\to_{\TpUniv} N\brk{x} : \TpUniv}
  }
  \and
  \ebrule[dfun-decoding]{
    \hypo{\Gamma\vdash M:\TpUniv}
    \hypo{\Gamma,x:\TpEl\prn{M}\vdash N\brk{x} : \TpUniv}
    \infer2{
      \Gamma\vdash
      \TpEl\prn[\big]{
        \prn{x:M}\to_{\TpUniv} N\brk{x}
      }
      =
      \prn{x:\TpEl\prn{M}}\to \TpEl\prn{N\brk{x}}\IsType
    }
  }
\end{mathparpagebreakable}

The decoding function $\TpEl\prn{-}$ and the subscripts $\to_{\TpUniv}$ are \emph{implicit} in most production systems; it is nonetheless important for the ``official'' rules of the type theory to be fully explicit, as it is with respect only to the explicit version that appropriate notions of \emph{model} can be specified. Implicit presentations of the calculus are then justified by proving metatheorems about the explicit syntax:

\begin{proposition}\label{prop:decoding-injective}
  In intensional Martin-L\"of type theory with one universe, the decoding function is injective in the sense that the following rule is \emph{admissible}:
  \begin{mathparpagebreakable}
    \ebrule[\textsc{univ-decoding-inj}]{
      \hypo{\Gamma\vdash M,N:\TpUniv}
      \hypo{\Gamma\vdash \TpEl\prn{M}= \TpEl\prn{N}\IsType}
      \infer[admissible]2{\Gamma\vdash M = N: \TpUniv}
    }
  \end{mathparpagebreakable}
\end{proposition}

It is the combination of \textsc{dfun-decoding} and \textsc{univ-decoding-inj} that justifies an implicit syntax in which $\TpEl$ is silent and $\to_{\TpUniv}$ is written as $\to$. The failure of \textsc{dfun-decoding} would lead to nonsensical error messages like \texttt{``$\prn{\_:\mathbb{N}}\to\mathbb{N}$  could not be unified with $\prn{\_:\mathbb{N}}\to\mathbb{N}$''}. The role of \textsc{univ-decoding-inj} is more subtle: we argue that the notional meaning of a \emph{silent} coercion from the universe to the class of all types must be an inclusion, and inclusions are always injective.

\subsection{Background: Type Theory with Multiple Universes}

Although most of pure mathematics can in principle be formalised without more than one universe~\citep{mclarty-2020}, it can be very inconvenient to do so. For this reason, most production implementations of dependent types have an entire hierarchy of universes $\TpUniv_0,\TpUniv_1,\ldots$, with each one encoded in the next. The idea of a universe hierarchy is simple enough, but there are so many possibilities for the exact formalism and its rules that no two systems agree.
Many of the notable discrepancies between different accounts of universe hierarchies come down to the question of \emph{cumulativity}, which asks:
\begin{quote}
  Given $i\leq j$ and term $\Gamma\vdash M:\TpUniv_i$, is there some term $\Gamma\vdash M':\TpUniv_j$ such that $\Gamma\vdash\TpEl_j\prn{M'}=\TpEl_i\prn{M}\IsType$?
\end{quote}

\subsubsection{Design of Non-Cumulative Hierarchies}
In systems like Lean and Agda, the hierarchy is explicitly \emph{non}-cumulative. In these systems, although it is possible to define a record type $\mathsf{Lift}_{i}^{j}\prn{M:\TpUniv_i}:\TpUniv_j$ such that $\TpEl_j\prn[\big]{\mathsf{Lift}_i^j\prn{M}}$ is \emph{isomorphic} to $\TpEl_i\prn{M}$, we would not have the judgemental equality that cumulativity requires. In the absence of cumulativity, it is not convenient to close \emph{each} universe under type connectives (like the dependent function space); instead, systems like Lean and Agda define the type constructors on the \emph{least upper bounds} of the universe levels of their components:
\begin{mathparpagebreakable}
  \ebrule[dfun-code-formation-lub]{
    \hypo{\Gamma\vdash M:\TpUniv_i}
    \hypo{\Gamma,x:\TpEl_i\prn{M}\vdash N\brk{x}:\TpUniv_j}
    \infer2{
      \Gamma\vdash \prn{x:M} \to_{i,j} N\brk{x}: \TpUniv_{i\sqcup j}
    }
  }
  \and
  \ebrule[dfun-decoding-lub]{
    \hypo{\Gamma\vdash M:\TpUniv_i}
    \hypo{\Gamma,x:\TpEl_i\prn{M}\vdash N\brk{x}:\TpUniv_j}
    \infer2{
      \Gamma\vdash
        \TpEl\prn[\big]{\prn{x:M} \to_{i,j} N\brk{x}}
        =
        \prn{x:\TpEl_i\prn{M}}\to \TpEl_j\prn{N\brk{x}}\IsType
    }
  }
\end{mathparpagebreakable}

Of course, our description does not perfectly match the implementations of Lean and Agda, which actually lack a separate top-level $\Gamma\vdash A\IsType$ judgement. We aim here only to draw out the important distinctions between different systems.

There are technical reasons why non-cumulative hierarchies are favoured in production systems:

\begin{enumerate}
\item Non-cumulative hierarchies are very easy to specify, and their rules are similar to those of \emph{Pure Type Systems}~\citep{barendregt-1992}.
\item Non-cumulative universes can avoid difficult implementation problems, including the contention between higher-order unification and the insertion of subtyping coercions.\footnote{For example, Agda has an experimental cumulative mode, but it works very poorly due to the interaction with unification.}
\item Non-cumulative hierarchies are compatible with a broad range of semantic models, including all Grothendieck $\infty$-toposes~\citep{shulman-2019}.
\end{enumerate}

On the other hand, most practitioners of pen-and-paper type theory continue to use cumulative universes, or assume (possibly erroneously) that anything they can write informally with a cumulative hierarchy can be formalised with a non-cumulative one. We believe that at least part of the gap between formal and informal practice can be explained by the current ``folklore'' status of both the metatheory of cumulative hierarchies \emph{and} the appropriate techniques for their practical implementation. We aim to bridge that gap here so that future implementers of dependently typed systems can assess the design space for universe hierarchies on informed grounds.

\subsubsection{Design of Cumulative Hierarchies}

In a cumulative hierarchy, a new ``lifting'' constructor is added that allows codes from one universe to be lifted to a higher one:
\begin{mathparpagebreakable}
  \ebrule[univ-lift-formation]{
    \hypo{\Gamma\vdash i \leq j}
    \hypo{\Gamma\vdash M:\TpUniv_i}
    \infer2{\Gamma\vdash \Lift{i}{j}{M} : \TpUniv_j}
  }
  \and
  \ebrule[univ-lift-decoding]{
    \hypo{\Gamma\vdash i \leq j}
    \hypo{\Gamma\vdash M:\TpUniv_i}
    \infer2{\Gamma\vdash\TpEl_j\prn{\Lift{i}{j}{M}} = \TpEl_i\prn{M}\IsType}
  }
\end{mathparpagebreakable}

Unfortunately, the equational rules for the lifting operation quickly become extremely complicated. First, we require lifting to be ``functorial'' in the preorder of universe levels in the sense that sequences of lifts can be collapsed:
\begin{mathparpagebreakable}
  \ebrule[univ-lift-func-idn]{
    \hypo{\Gamma\vdash M:\TpUniv_i}
    \infer1{\Gamma\vdash \Lift{i}{i}{M} = M : \TpUniv_i}
  }
  \and
  \ebrule[univ-lift-func-cmp]{
    \hypo{\Gamma\vdash i \leq j \leq k}
    \hypo{\Gamma\vdash M:\TpUniv_i}
    \infer2{\Gamma\vdash \Lift{j}{k}\Lift{i}{j}M = \Lift{i}{k}M : \TpUniv_k}
  }
\end{mathparpagebreakable}

In addition, we require a system of ``naturality'' rules that ensure that the lifting operation commutes with all built-in type constructors. For example:
\begin{mathparpagebreakable}
  \ebrule[univ-lift-dfun-natural]{
    \hypo{\Gamma\vdash i\leq j}
    \hypo{\Gamma\vdash M:\TpUniv_i}
    \hypo{\Gamma,x:\TpEl_i\prn{M}\vdash N\brk{x}\colon \TpUniv_i}
    \infer3{
      \Gamma\vdash
      \Lift{i}{j}\prn[\big]{
        \prn{x:M}\to_i N\brk{x}
      }
      =
      \prn{x:\Lift{i}{j}{M}}\to \Lift{i}{j}N\brk{x}
      : \TpUniv_j
    }
  }
\end{mathparpagebreakable}

We will observe that a counterpart to \zcref{prop:decoding-injective} holds for a cumulative universe hierarchy satisfying the described laws:

\begin{proposition}\label{prop:hierarchy-decoding-injective}
  In Martin-L\"of type theory with a cumulative universe hierarchy, each decoding function is injective in the sense that the following rule is admissible:
  \begin{mathparpagebreakable}
    \ebrule[\textsc{univ-decoding-inj}]{
      \hypo{\Gamma\vdash M,N:\TpUniv_i}
      \hypo{\Gamma\vdash \TpEl_i\prn{M}= \TpEl_i\prn{N}\IsType}
      \infer[admissible]2{\Gamma\vdash M = N: \TpUniv_i}
    }
  \end{mathparpagebreakable}
\end{proposition}

We stress that the admissibility of \textsc{univ-decoding-inj} is extremely sensitive to the presence of naturality laws like \textsc{univ-lift-dfun-natural}: if any one of them were omitted, \zcref{prop:hierarchy-decoding-injective} would surely fail. When the necessary laws are included, however, it furthermore follows that the lifting operations too are injective.

\begin{corollary}\label{cor:univ-lift-inj}
  In Martin-L\"of type theory with a cumulative universe hierarchy, each lifting function is injective in the sense that the following rule is admissible:
  \begin{mathparpagebreakable}
    \ebrule[\textsc{univ-lift-inj}]{
      \hypo{\Gamma\vdash i \leq j}
      \hypo{\Gamma\vdash M,N:\TpUniv_i}
      \hypo{\Gamma\vdash \Lift{i}{j}M= \Lift{i}{j}N:\TpUniv_j}
      \infer[admissible]3{\Gamma\vdash M = N: \TpUniv_i}
    }
  \end{mathparpagebreakable}
\end{corollary}

\begin{proof}
  Suppose we have $\Lift{i}{j}M= \Lift{i}{j}{N}$. By congruence of $\TpEl_i$ we therefore have $\TpEl_j\prn{\Lift{i}{j}M}= \TpEl_j\prn{\Lift{i}{j}{N}}$, which by \textsc{univ-lift-decoding} is equivalent to $\TpEl_i\prn{M}=\TpEl_i\prn{N}$, from which we conclude by \zcref{prop:hierarchy-decoding-injective} that $M= N$.
\end{proof}

\subsection{Key Ideas: the Fuss-Free Approach to Cumulative Universes}

In light of \zcref{prop:hierarchy-decoding-injective,cor:univ-lift-inj}, we might revisit the design of the cumulative hierarchy. If we instead took \textsc{univ-decoding-inj} as a primitive (derivable) rule, the functoriality and naturality rules for the lifting operation would become redundant!  For example:

\begin{enumerate}
  \item \textsc{univ-lift-func-cmp}: To deduce  $\Lift{j}{k}\Lift{i}{j}M = \Lift{i}{k}M$, it suffices by \textsc{univ-decoding-inj} to show that $\TpEl_k\prn{\Lift{j}{k}\Lift{i}{j}M}= \TpEl_k\prn{\Lift{i}{k}{M}}$ which follows from a few applications of \textsc{univ-lift-decoding}.
  \item \textsc{univ-lift-dfun-natural}: To deduce $\Lift{i}{j}\prn[\big]{\prn{x:M}\to_i N\brk{x}}=\prn{x:\Lift{i}{j}{M}}\to \Lift{i}{j}N\brk{x}$, it suffices by \textsc{univ-decoding-inj} to show that $\TpEl_j\prn[\Big]{\Lift{i}{j}\prn[\big]{\prn{x:M}\to_i N\brk{x}}}=\TpEl_j\prn[\big]{\prn{x:\Lift{i}{j}{M}}\to \Lift{i}{j}N\brk{x}}$, which follows from a combination of \textsc{univ-lift-decoding} and \textsc{dfun-decoding}.
\end{enumerate}

Replacing all functoriality and naturality rules with a \emph{single} injectivity law is already a significant simplification, but we can go further. At the moment, our syntax must have constructors for type connectives at every universe level which duplicate the corresponding constructors at the type level; this duplication must then be mediated by decoding laws like \textsc{dfun-decoding}. We can mitigate the duplication, however, by introducing a \emph{judgemental notion of smallness} (\zcref{sec:judgemental-smallness}) that axiomatises the image of the decoding maps.

\subsubsection{A Judgemental Notion of Smallness}\label{sec:judgemental-smallness}

Our final simplification is to take the injectivity of decoding seriously, and simply replace $\TpUniv_i$ with its \emph{\textbf{image}} in the class of all types. In type theory, this amounts to defining a judgemental notion of \emph{smallness} and then using this to give uniform introduction and elimination rules to each universe. In particular, we define a new form of judgement $\Gamma\vdash A\IsSmall_i$ presupposing $\Gamma\vdash A\IsType$ and close it under rules that express the closure properties of the type connectives:\footnote{Smallness is here formulated as a ``synthetic'' judgement in the sense of \citet{martin-lof-1994}. In particular, $A$ is not the \emph{subject} of the judgement $A\IsSmall_i$ but rather one of its presuppositions. We could (and later shall) equivalently formulate smallness analytically as $p::A\IsSmall_i$ with subject $p$ ranging over smallness witnesses, with an additional rule asserting $p= q::A\IsSmall_i$ for all $p,q$. In an implementation, the synthetic formulation is ideal because there is no need to carry around witnesses that can be decidably reconstructed on demand.}
\begin{mathparpagebreakable}
  \ebrule[dfun-small]{
    \hypo{\Gamma\vdash A\IsSmall_i}
    \hypo{\Gamma,x:A\vdash B\brk{x}\IsSmall_i}
    \infer2{\Gamma\vdash \prn{x:A}\to B\brk{x}\IsSmall_i}
  }
  \and
  \ebrule[univ-small]{
    \hypo{\Gamma\vdash i < j}
    \infer1{\Gamma\vdash \TpUniv_i\IsSmall_j}
  }
  \and
  \ebrule[small-mono]{
    \hypo{\Gamma\vdash i \leq j}
    \hypo{\Gamma\vdash A\IsSmall_i}
    \infer2{\Gamma\vdash A\IsSmall_j}
    }
\end{mathparpagebreakable}

Then introduction and elimination rules for $\TpUniv_i$ are given with respect to the smallness judgement, satisfying appropriate $\beta$- and extensionality laws:\footnote{Of course, the extensionality law is equivalent to an $\eta$-law.}
\begin{mathparpagebreakable}
  \ebrule[univ-intro]{
    \hypo{\Gamma\vdash A\IsSmall_i}
    \infer1{\Gamma\vdash\Shrink{i}{A} : \TpUniv_i}
  }
  \and
  \ebrule[univ-elim]{
    \hypo{\Gamma\vdash M:\TpUniv_i}
    \infer1{\Gamma\vdash \Enlarge{i}{M}\IsType}
  }
  \and
  \ebrule[univ-elim-small]{
    \hypo{\Gamma\vdash M:\TpUniv_i}
    \infer1{\Gamma\vdash \Enlarge{i}{M}\IsSmall_i}
  }
  \and
  \ebrule[univ-beta]{
    \hypo{\Gamma\vdash A\IsSmall_i}
    \infer1{\Gamma\vdash \Enlarge{i}{\Shrink{i}{A}} = A\IsType}
  }
  \and
  \ebrule[univ-ext]{
    \hypo{\Gamma\vdash \Enlarge{i}{M} = \Enlarge{i}{N}\IsType}
    \infer1{\Gamma\vdash M= N:\TpUniv_i}
  }
\end{mathparpagebreakable}

In fact, our rules for universes are parametric in \emph{any} notion of smallness and therefore can be made to modularly accommodate a variety of systems in both theory and implementation.

Because it is easy to specify, easy to implement, and (we contend) easy to use, we refer to the calculus with judgemental smallness as the \textbf{fuss-free approach to universes}.

\subsubsection{Syntactic Equivalence of Fuss-Free Universes}
We will show that, up to admissibility, the simplified calculus with judgemental smallness is \emph{equivalent} to the more complicated cumulative hierarchy with functorial lifting operations and naturality laws. Indeed, the lifting operation can be implemented as $\Lift{i}{j}{M}:= \Shrink{j}{\Enlarge{i}{M}}$ and its entire equational theory is derivable from the \textsc{univ-beta} and \textsc{univ-ext}. This has consequences both for theory and practice:

\begin{enumerate}
  \item \textbf{Theory}: We can interpret the very convenient type theory with judgemental smallness into semantic models of cumulative hierarchies \emph{without} requiring the model to satisfy injectivity laws, as these can first be eliminated via our syntactic equivalence.
  \item \textbf{Practice}: The smallness calculus offers non-trivial simplifications for implementation: many equational rules and generating syntactic forms (which are necessarily duplicated across the syntax and the semantic domain) are replaced by a very small collection of easily programmed rules. Furthermore, a direct implementation of the functorial lifting coercions may require techniques like those of \citet{allais-mcbride-boutillier-2013}, whereas the smallness calculus can be implemented by adapting a stock algorithm like that of \citet{coquand-type-checking} as we shall show.
\end{enumerate}

\subsubsection{Cumulative Inductive Types in the Fuss-Free Style: Staying Large}\label{sec:cumulative-inductive-types}

For several years, the Rocq proof assistant has had experimental modes for universe polymorphism and universe cumulativity. Inductive datatypes were an early sticking point in the interaction between these two features, as \citet{timany-jacobs-2015} pointed out; the problem is easiest to illustrate for a universe polymorphic list datatype constructor, where given levels $i\leq j$ one could have $t:\mathsf{List}_i~A$ but not $t:\mathsf{List}_j~A$ unless $i=j$. This problem was solved by \citet{timany-sozeau-2018}, who proposed to have subtypings $\mathsf{List}_i~A\leq\mathsf{List}_j~A$ regardless of how $i$ and $j$ are related in the level order.

The subtyping relation proposed by \citet{timany-sozeau-2018} is perfectly sensible, and it also appears in our own system (up to the insertion of $\Enlarge{}/\Shrink{}$ coercions). But in the fuss-free style, we may go a little further and avoid parameterising datatypes ahead of time in universe levels. Indeed, excessive level parameterisation quickly becomes unwieldy and may even lead to serious problems at library boundaries~\citep{int-tree-issues}.

Our methodology starts from the following principle: \textbf{stay large}. In other words, rather than adding level parameters everywhere to anticipate the closure of universes under a defined type constructor, the constructor should instead be defined on \emph{types}; then, when it needs to be passed as an element of a universe, the fuss-free rules allow it to be shrunk appropriately. We believe that ``staying large'' is most in keeping with standard mathematical practice, where (\emph{e.g.}) the concept of a category is prior to the concept of a \emph{locally small} category. Returning to lists, in a fuss-free system we would define the list type constructor without reference to universes at all:
\begin{lstlisting}[basicstyle=\ttfamily\small]
data List A where {
   nil : List A;
   cons : A -> List A -> List A
};
\end{lstlisting}

The declaration above would correspond in a proper implementation to the definitional extension of the theory by a new datatype constructor
\begin{mathpar}
  \ebrule{
    \hypo{A\ \mathit{type}}
    \infer1{\mathsf{List}~A\ \mathit{type}}
  }
  \and
  \ebrule{
    \hypo{A\ \mathit{type}}
    \infer1{\mathsf{nil}_A : \mathsf{List}~A}
  }
  \and
  \ebrule{
    \hypo{A\ \mathit{type}}
    \hypo{x:A}
    \hypo{\mathit{xs}:\mathsf{List}~A}
    \infer3{\mathsf{cons}_A~x~\mathit{xs} : \mathsf{List}~A}
  }
\end{mathpar}
whose instances $\mathsf{List}\prn{A}$ may all be implemented in terms of the fixed point of an appropriate datatype description in the style of \citet{dagand:2013}. The rules of the fuss-free system would then, however, ensure that when $A$ is small, so is $\mathsf{List}~A$. We are \emph{able} to define a polymorphic constructor
\begin{align*}
  \mathsf{list}_i &\coloneq \lambda X:\TpUniv_i\mathpunct{.} \Shrink{i}\prn[\big]{\mathsf{List}~\prn{\Enlarge{i}X}}
\end{align*}
with the property that for each $i\leq j$ and $M:\TpUniv_i$ we have $\Enlarge{i}\mathsf{list}_i\prn{M} = \Enlarge{j}\mathsf{list}_j\prn{\Shrink{j}\Enlarge{i}M}$, but the real lesson is that we don't have to. Instead, we keep $\mathsf{List}$ large and shrink it down only when needed; because the $\Enlarge{}/\Shrink{}$ coercions can easily be inserted automatically by an elaborator in a syntax-directed way, we may simply write $\mathsf{List}$ everywhere without bothering with levels.

Our proposal will work for any inductive type that can be encoded by a datatype description, which we adapt from \citet{dagand:2013}. Unlike \emph{op.\ cit.}, we have a top-level judgemental notion of datatype description which we then instrument, in fuss-free fashion, with a hierarchy of universes of small datatype descriptions. Instead of universe levels being an intrinsic part of the datatype definition~\citep{timany-sozeau-2018}, we describe the shape of a datatype without regard to size except where absolutely necessary; size then emerges naturally at the \emph{use site} instead of the \emph{definition site}.

\subsection{History and Related work}

\subsubsection{Structural \emph{vs.}\ Material Perspective}\label{sec:structural-vs-material}
We have so far adopted a ``structural'' or \emph{generalised algebraic} perspective on type theory in which we treat different form of judgement as ranging over disjoint syntactic categories (\emph{e.g.}\ types and terms). We allow that these distinctions may be elided where justified by a coherent system of implicit coercions, but we have stressed that the official rules are given in a structural style that directly evinces the universal property of the initial categorical model of type theory~\citep{dybjer-1996,uemura-2021-thesis,bcde-2021}.
An alternative historical approach is to treat the embedding of (certain) terms into types not as a coercion but as an \emph{actual} subset inclusion, which (by analogy with the contrast between structural and material set theory~\citep{shulman:2019:set-theory}) we might refer to as a \emph{material} perspective on type theory.

The advantage of structural formulations is that they give rise (automatically!) to a general notion of model that makes sense in any category. This is important: present-day applications of type theory to both computer science and general mathematics rely on constructing models not only in the category of sets, but also in other toposes~\citep{reus-1995,bmss-2011,sterling-harper:2021,cherubini-coquand-hutzler-2024,grodin-li-harper-2026} and even quasi-toposes~\citep{kammar-2021-pihoc}; in these cases, a structural approach is not just desirable but required. Conversely, even the \emph{syntactic metatheory} (normalisation, decidability, \emph{etc.}) calls for a structural notion of model in which Tait's method of computability can be phrased in terms of Artin glueing~\citep{fiore:2002,altenkirch-kaposi-2016-nbe,sterling-angiuli-2021,BocquetPhD}.

In contrast, the material interpretation of universes evinces no single general notion of model, although specific \emph{interpretations} do exist: there is Barras' set-theoretic interpretation~\citep{barras-2012} in which decoding is a subset inclusion, and then there are realisability-style models~\citep{allen-1987-lics} in which decoding is a PER inclusion. In spite of these examples, there is no general theory of model for the material formulation.
Moreover, the viability of the material perspective can be doubted in light of its failure to justify the obvious subtyping rule for dependent function types~\citep{barras-2012,timany-sozeau-2018}:
\begin{mathparpagebreakable}
  \ebrule[dfun-subtyping*]{
    \hypo{\Gamma\vdash A' <: A}
    \hypo{\Gamma,x:A'\vdash B[x] <: B'[x]}
    \infer2{
      \Gamma\vdash \prn{x:A}\to B[x] <: \prn{x:A'}\to B'[x]
    }
  }
\end{mathparpagebreakable}

The rule above presents no difficulty for structural accounts of type theory where subtyping inclusions are treated as (possibly implicit) coherent coercions~\citep{luo-1997}. For this reason, we advocate for the structural interpretation, for which we propose a simpler formulation.

\subsubsection{Universes \`a la Russell and \`a la Tarski}\label{sec:russell-vs-tarski}

In the case of universes, the material/structural divide (\zcref{sec:structural-vs-material}) is often described as the contrast between universes cooked ``\`a la Russell'' and ``\`a la Tarski'', where the latter is understood to involve an explicit decoding family that the former lacks. We have recently become aware that this folk interpretation of Per Martin-L\"of's contrast between Russellian and Tarskian universes~\citep{martin-lof-1984} is at odds with his original intent~\citep{dybjer-2025-universes}---in particular, Martin-L\"of actually meant to draw attention to the contrast between \emph{open} and \emph{closed} universes, and it seems to have been a coincidence that his exposition was notationally aligned with the choice of implicit \emph{vs.}\ explicit decoding.
We can clarify the point slightly:

\begin{enumerate}
  \item The eliminator of an open/Russellian universe is the (implicit \emph{or} explicit) decoding family.
  \item The eliminator of a closed/Tarskian universe is the \emph{induction-recursion} principle that expresses its universal property as the smallest family of types closed under certain operations.
\end{enumerate}

Our work primarily concerns the comparison of implicit and explicit presentations of Russellian (open) universes from the structural point of view, with direct applications to the implementation of dependently typed programming languages and interactive proof assistants.

Although these two representations have long been presumed equivalent, formally establishing this connection is highly technical.
\citet{Assaf} proposed the first proof, which contained errors that were later corrected by \citet{felicissimo-winterhalter}, who relied on confluence techniques that do not adapt to theories with $\eta$-equivalence.
Recently, \citet{coquand-sterbac-2026} presented a formalised proof supporting $\eta$-conversion, relying solely on injectivity of constructors. 

\subsubsection{Universes \`a la Coquand}\label{sec:coquand-universes}

To the untrained eye, our proposal (\zcref{sec:judgemental-smallness}) may appear similar to that of \citet{coquand-2012-presheaf-model}, but there are some substantive differences:
\begin{enumerate}
  \item Coquand's theory is not \emph{in fact} generalised algebraic\footnote{Indeed, the universe level $n$ is not an argument to the type constructor $\Pi~A~B\ \mathit{type}_n$ in \citet{coquand-2012-presheaf-model}; the presheaf interpretation of \emph{op.\ cit.}\ rests, as a result, on the set-theoretic arcana of presheaves rather than their categorical structure.} and therefore does not give rise to an automatic definition of model (but see \citet{coquand-sterbac-2026}).
  \item Unlike Coquand, we retain the judgemental/syntactic distinction between types and terms, which (in the algebraic setting) enables many level annotations to be erased by decoding.
  \item Coquand uses external universe levels and cannot support level polymorphism.
\end{enumerate}

We note that several others have given a generalised algebraic version of \emph{universes \`a la Coquand}; see for example the work of \citet{gratzer-kavvos-birkedal-2021}.
We claim that our proposal is, in contrast, the closest thing possible to a structural account of Martin-L\"of's \emph{universes \`a la Russell} (\zcref{sec:russell-vs-tarski}) extended by universe polymorphism.

\subsubsection{Explicit Universe Polymorphism}
To avoid requiring a programmer to write the same code at multiple universe levels, most systems include some kind of polymorphism over universe levels. Recently, \citet{bcde-2022} proposed a structural account of several variations on universe polymorphism, which are currently in the process of being implemented in Rocq~\citep{bezem-sozeau-2025}. Our work is complementary to that of \emph{op.\ cit.}, and we prove normalisation for the cumulative variant of this theory (\zcref{sec:normalisation}).
In fact, the fuss-free approach is designed to be compatible with \emph{any} other reasonable form of universe polymorphism. Notably, a first normalisation result was recently proven and formalised by \citet{dfk-2026} for a non-cumulative theory with \emph{first-class universe levels}, \emph {i.e.}\ where levels form a type like in Agda.

\subsubsection{Typical Ambiguity and Type Checking}
A related desideratum is Feferman's ``typical ambiguity''~\citep{feferman-2004}, which allows programmers to leave out universe levels entirely and have the system infer them---either by means of ML-style implicit polymorphism~\citep{harper-pollack-1991,sozeau-tabareau-2014} or by means of a syntactic displacement operation as in Mc~Bride's \emph{Crude but Effective Stratification}~\citep{mcbride-2002-crude} (see also Favonia et al.~\citep{favonia-angiuli-mullanix-2023}). Admittedly, our attitude on the practical importance of typical ambiguity is one of light skepticism, but our work will again be complementary to any reasonable approach.

In section \zcref{sec:elaboration}, we describe a type checking algorithm for the fuss-free hierarchy.
We differ from the works of \citet{harper-pollack-1991,herbelin-2005} in that we do not try to infer universe levels. Instead, we aim to make working with universe levels as simple as possible, while minimising the number of places where they need to appear.

\subsubsection{Elaboration of Datatype Descriptions}
We describe a fuss-free approach for cumulative inductive types in \zcref{sec:fuss-free-datatypes}, which builds upon Dagand's datatype descriptions~\citep{dagand:2013}.
In a practical system, one would implement the elaboration algorithm from user-facing inductive definitions to the underlying internal description.
This was one of the core aspects of Epigram~2, based on the work of \citet{levitation}, and also included in Atkey's Foveran~\citep{foveran}.

\subsection{Structure of this Paper}
\begin{itemize}
  \item[\bfseries\S\S~\ref{sec:natural-hierarchies}--\ref{sec:fuss-free}] We present and compare two algebraic presentations of cumulative universes with explicit universe polymorphism: a \emph{natural hierarchy} (\ThyMLTTNat{}) with labelled codes at every level satisfying naturality laws with respect to explicit level coercions, and a \emph{fuss-free hierarchy} (\ThyFussFree{}) with judgemental smallness and universes defined by an introduction rule, elimination rule, and $\beta$-/$\eta$-laws.
  \item[\bfseries\S~\ref{sec:fuss-free-datatypes}] We describe how to handle cumulative inductive types~\citep{timany-sozeau-2018} in \ThyFussFree{} by introducing a judgemental notion of datatype description as well as fuss-free universes of datatype descriptions.
  \item[\bfseries\S~\ref{sec:normalisation}] We prove normalisation for \ThyMLTTNat{}, as suggested by \citet{bcde-2022}. Normalisation implies the admissibility of injectivity laws for universe decoding and lifting in \ThyMLTTNat{}, from which we deduce that \ThyMLTTNat{} and \ThyFussFree{} are equivalent up to syntactic translation.
  \item[\bfseries\S\S~\ref{sec:elaboration}--\ref{sec:tutorial-implementation}] We describe a bidirectional type checking algorithm for the fuss-free calculus, which we then implement in Haskell. Our implementation incorporates the elaboration of user-written datatype declarations to descriptions in the style of \citet{dagand:2013}, as well as record types.
\end{itemize}

We note that our fuss-free calculus is already impacting the design of new implementations of type theory; for example, \citet{lynch-2026} uses the fuss-free approach to seamlessly introduce a static/dynamic phase distinction~\citep{sterling-harper:2021} only at small types.
\section{Martin-L\"of Type Theory and Natural Universe Hierarchies}\label{sec:natural-hierarchies}

We will describe here an ``orthodox'' presentation of a version of Martin-L\"of type theory (\ThyMLTT) that will serve as the basis for two presentations of universes that we will later show to be equivalent up to admissibility.  Although we have in \zcref{sec:introduction} for familiarity's sake used a traditional notation to informally describe the rules of type theory, we shall now switch to a more precise ``logical framework'' presentation in the doctrine of \emph{second-order generalised algebraic theories}~\citep{uemura-2021-thesis,kaposi-xie-2024}.

A second-order generalised algebraic theory is defined by setting down sorts, operations, and equations. Unlike in an ordinary algebraic theory, the sorts are \emph{dependent} and the operations are allowed to bind variables by means of a stratified dependent function space. A sort whose variables are allowed to be bound is called \emph{representable}. A \emph{context} in the sense of type theory is then defined to be a telescope in which every binding has a representable sort. The basic judgemental structure of Martin-L\"of type theory is described, then, by specifying a non-representable sort of types and a representable \emph{dependent} sort of elements:

\begin{rulebox}[title = {\ThyMLTTZero{} $\gg$ \emph{basic forms of judgment}}]
  \begin{mathparsmall}
    \ebrule{
      \infer0{\SortTp : \SORT}
    }
    \and
    \ebrule{
      \hypo{A:\SortTp}
      \infer1{\SortTm(A):\REPSORT}
    }
  \end{mathparsmall}
\end{rulebox}

\ThyMLTT{} then extends \ThyMLTTZero{} by the standard top-level type connectives; for brevity's sake, we focus on dependent function types, but other standard type constructives (like dependent pair types, the unit type, identity types, \emph{etc.}) can be handled analogously.

\begin{rulebox}[title = {\ThyMLTT{} $\gg$ \emph{dependent function types}}]
  \begin{mathparsmall}
    \ebrule{
      \hypo{A : \SortTp}
      \hypo{B : \SortTm\prn{A}\to\SortTp}
      \infer2{\TpPi\prn{A,B}: \SortTp}
    }
    \and
    \ebrule{
      \hypo{A : \SortTp}
      \hypo{B : \SortTm\prn{A}\to \SortTp}
      \hypo{M : \prn{x : \SortTm\prn{A}}\to \SortTm\prn{B\prn{x}}}
      \infer3{\TmLam\brc{A,B}\prn{M}:\SortTm\prn{\TpPi\prn{A,B}}}
    }
    \and
    \ebrule{
      \hypo{A : \SortTp}
      \hypo{B : \SortTm\prn{A}\to \SortTp}
      \hypo{M : \SortTm\prn{\TpPi\prn{A,B}}}
      \hypo{N : \SortTm\prn{A}}
      \infer4{\TmApp\brc{A,B}\prn{M,N} : \SortTm\prn{B\prn{N}}}
    }
    \and
    \ebrule{
      \hypo{A : \SortTp}
      \hypo{B : \SortTm\prn{A}\to \SortTp}
      \hypo{M : \prn{x : \SortTm\prn{A}}\to \SortTm\prn{B\prn{x}}}
      \hypo{N : \SortTm\prn{A}}
      \infer4{
        \TmApp\brc{A,B}\prn{
          \TmLam\brc{A,B}\prn{M},
          N
        }
        =
        M\prn{N}: \SortTm\prn{B\prn{N}}
      }
    }
    \and
    \ebrule{
      \hypo{A : \SortTp}
      \hypo{B : \SortTm\prn{A}\to \SortTp}
      \hypo{M : \SortTm\prn{\TpPi\prn{A,B}}}
      \infer3{
        M = \TmLam\brc{A,B}\prn{
          \lambda x\mathpunct{.}
          \TmApp\brc{A,B}\prn{M,x}
        } :
        \SortTm\prn{\TpPi\prn{A,B}}
      }
    }
  \end{mathparsmall}
\end{rulebox}

For lack of space, we omit entirely standard rules for dependent sums $\Sigma\prn{A,B}$ and the unit type.

\subsection{A Judgemental Preorder of Universe Levels}

We will consider universes that are parameterised in a preorder of levels; we will make the sort of levels \emph{representable} so that level variables may appear in the context. The theory \ThyLevels{} of universe levels is introduced below:

\begin{rulebox}[title = {\ThyLevels{} $\gg$ \textit{judgemental levels preorder}}]
  \begin{mathparsmall}
    \ebrule{
      \infer0{\SortLvl:\REPSORT}
    }
    \and
    \ebrule{
      \hypo{i,j:\SortLvl}
      \infer1{i\leq j : \SORT}
    }
    \and
    \ebrule{
      \hypo{i,j:\SortLvl}
      \hypo{p,q: i \leq j}
      \infer2{p = q : i \leq j}
    }
    \and
    \ebrule{
      \hypo{i:\SortLvl}
      \infer1{\_:i\leq i}
    }
    \and
    \ebrule{
      \hypo{i,j,k:\SortLvl}
      \hypo{\_:i \leq j}
      \hypo{\_:j\leq k}
      \infer3{\_:i\leq k}
    }
    \and
    \ebrule{
      \infer0{0:\SortLvl}
    }
    \and
    \ebrule{
      \hypo{i:\SortLvl}
      \infer1{i^+ : \SortLvl}
    }
    \and
    \ebrule{
      \hypo{i:\SortLvl}
      \infer1{\_:0 \leq i}
    }
    \and
    \ebrule{
      \hypo{i:\SortLvl}
      \infer1{\_: i \leq i^+}
    }
    \and
    \ebrule{
      \hypo{i,j:\SortLvl}
      \hypo{\_:i \leq j}
      \infer2{\_: i^+ \leq j^+}
    }
  \end{mathparsmall}
\end{rulebox}

By virtue of the third rule above establishing the \emph{judgemental proof irrelevance} of the level inequality sort, we have felt free to write the rules in an implicit form with underscores. Formally, there are of course explicit named generators and witnesses for these but because the choice does not matter, we may leave them out as a matter of notation.
We can add additional operations on levels if we like, including least upper bounds, but we will leave it alone.

\subsection{Explicit Universe Polymorphism}

Finally we add explicit universe polymorphism in the style of \citet{bcde-2022}.\footnote{For the sake of simplicity, we do not add constraint-indexed products.} We introduce $\ThyLevelPoly$ extending $\ThyMLTT_0\cup\ThyLevels$ below:

\begin{rulebox}[title = {\ThyLevelPoly{} $\gg$ \textit{explicit universe polymorphism}}]
  \begin{mathparsmall} 
    \ebrule{
      \hypo{A : \SortLvl \rightarrow \SortTp}
      \infer1{\TpLvlPi\prn{A} : \SortTp}
    }
    \and
    \ebrule{
      \hypo{A:\SortLvl\to \SortTp}
      \hypo{M:\prn{\alpha:\SortLvl}\to \SortTm\prn{A\prn{\alpha}}}
      \infer2{
        \TmLvlLam\brc{A}\prn{M} : \SortTm\prn{\TpLvlPi\prn{A}}
      }
    }
    \and
    \ebrule{
      \hypo{A:\SortLvl\to\SortTp}
      \hypo{M:\SortTm\prn{\TpLvlPi\prn{A}}}
      \hypo{i:\SortLvl}
      \infer3{
        \TmLvlApp\brc{A}\prn{M,i} : \SortTm\prn{A\prn{i}}
      }
    }
    \and
    \ebrule{
      \hypo{A:\SortLvl\to \SortTp}
      \hypo{M:\prn{\alpha:\SortLvl}\to \SortTm\prn{A\prn{\alpha}}}
      \hypo{i:\SortLvl}
      \infer3{
        \TmLvlApp\brc{A}\prn{\TmLvlLam\brc{A}\prn{M},i}
        =
        M\prn{i}
        :
        \SortTm\prn{A\prn{i}}
      }
    }
    \and
    \ebrule{
      \hypo{A:\SortLvl\to\SortTp}
      \hypo{M:\SortTm\prn{\TpLvlPi\prn{A}}}
      \infer2{
        M = \TmLvlLam\brc{A}\prn{\lambda \alpha\mathpunct{.}\TmLvlApp\brc{A}\prn{M,\alpha}}
      }
    }
  \end{mathparsmall}
\end{rulebox}

\subsection{Optional Extension: a Transfinite Universe}

This is only the most basic possible theory of level polymorphism. In \zcref{sec:cumulative-inductive-types} (``staying large''), we argued that parameterised data types could be introduced whose type constructors may take types as arguments (\emph{e.g.}\ the `list' type). Such an extension is perfectly reasonable, but its implementation is complicated by the fact that parameterised type constructors become \emph{schematic} in their parameters rather than expressible as a single family of types. One simplification, which we have carried out in our implementation, is to introduce a single ``big'' level that is strictly larger than all the finite levels, making level variables range only over finite levels, and retaining $(-)^+$ as an operation on finite levels only. The introduction of a transfinite universe does not significantly change the structural properties of the language (except insofar as the adapted normalisation proof will require a transfinite universe in the metatheory).

\subsection{Natural Universe Hierarchies}

We now describe the extension of Martin-L\"of type theory by a functorial universe lifting coercion that commutes with all connectives; we refer to this extension as a \emph{natural universe hierarchy} because the commutation with connectives can be seen to be a naturality law from an appropriate perspective. Formally, we define the signature \ThyMLTTNat{} extending $\ThyMLTT{} \cup \ThyLevelPoly{}$ as follows:.

\begin{rulebox}[title = {\ThyMLTTNat{} $\gg$ \textit{universe types, decoding, and lifting}}]
  \begin{mathparsmall}
    \ebrule{
      \hypo{i:\SortLvl}
      \infer1{\strut\TpUniv_i:\SortTp}
    }
    \and
    \ebrule{
      \hypo{i:\SortLvl}
      \hypo{M:\SortTm\prn{\TpUniv_i}}
      \infer2{
        \strut \TpEl_i(M): \SortTp
      }
    }
    \and
    \ebrule{
      \hypo{i,j:\SortLvl}
      \hypo{\_: i\leq j}
      \hypo{M:\SortTm\prn{\TpUniv_i}}
      \infer3{\strut\Lift{i}{j}{M} : \SortTm\prn{\TpUniv_j}}
    }
    \and
    \ebrule{
      \hypo{i,j:\SortLvl}
      \hypo{\_: i\leq j}
      \hypo{M:\SortTm\prn{\TpUniv_i}}
      \infer3{\TpEl_j\prn[\big]{\Lift{i}{j}{M}} = \TpEl_i\prn{M} : \SortTp}
    }
    \and
    \ebrule{
      \hypo{i:\SortLvl}
      \hypo{M:\SortTm\prn{\TpUniv_i}}
      \infer2{
        \strut\Lift{i}{i}{M}= M : \SortTm\prn{\TpUniv_i}
      }
    }
    \and
    \ebrule{
      \hypo{i,j,k:\SortLvl}
      \hypo{\_:i\leq j}
      \hypo{\_:j\leq k}
      \hypo{M:\SortTm\prn{\TpUniv_i}}
      \infer4{
        \strut\Lift{j}{k}{\Lift{i}{j}{M}}= \Lift{i}{k}{M} : \SortTm\prn{\TpUniv_k}
      }
    }
  \end{mathparsmall}
\end{rulebox}

Next we must add \emph{codes} inside each universe for the universes themselves as well as other type connectives (including dependent function types).

\begin{rulebox}[title = {\ThyMLTTNat{} $\gg$ \textit{codes for types in universes}}]
\begin{mathparsmall}
  \ebrule{
    \hypo{i,j:\SortLvl}
    \hypo{\_:i^+\leq j}
    \infer2{
      \CodeUniv_i^j : \SortTm\prn{\TpUniv_j}
    }
  }
  \and
  \ebrule{
    \hypo{M:\SortTm\prn{\TpUniv_i}}
    \hypo{N:\SortTm\prn{\TpEl_i\prn{M}} \to \SortTm\prn{\TpUniv_i}}
    \infer2{
      \CodePi_i\prn{M,N} : \SortTm\prn{\TpUniv_i}
    }
  }
  \and
  \ebrule{
    \hypo{M:\SortLvl\to \SortTm\prn{\TpUniv_i}}
    \infer1{
      \CodeLvlPi^i\prn{M}:\SortTm\prn{\TpUniv_i}
    }
  }
  \and
  \ebrule{
    \hypo{i,j:\SortLvl}
    \hypo{\_:i^+\leq j}
    \infer2{
      \TpEl_j\prn[\big]{\CodeUniv_i^j} = \TpUniv_i : \SortTp
    }
  }
  \and
  \ebrule{
    \hypo{M:\SortTm\prn{\TpUniv_i}}
    \hypo{N:\SortTm\prn{\TpEl_i\prn{M}} \to \SortTm\prn{\TpUniv_i}}
    \infer2{
      \TpEl_i\prn{
        \CodePi_i\prn{M,N}
      }=
      \TpPi\prn{
        \TpEl_i\prn{M},
        \lambda x\mathpunct{.}
        \TpEl_i\prn{N\prn{x}}
      }:\SortTp
    }
  }
  \and
  \ebrule{
    \hypo{M:\SortLvl\to\SortTm\prn{\TpUniv_i}}
    \infer1{
      \TpEl_i\prn[\big]{\CodeLvlPi^i\prn{M}}
      =
      \TpLvlPi\prn{\lambda x\mathpunct{.}\TpEl_i\prn{M\prn{x}}}
    }
  }
\end{mathparsmall}
\end{rulebox}

The rules above express that the type encoded by $\CodePi_i\prn{M,N}$ is a $\Pi$-type, and that the type encoded by $\CodeUniv_i^j$ is the universe $\TpUniv_i$. We shall \emph{also} add rules that make the lifting operation $\Lift{i}{h}$ preserve the type codes; this can be seen to correspond to a \emph{naturality} law in a categorical formulation of the universe hierarchy, whence the name.

\begin{rulebox}[title = {\ThyMLTTNat{} $\gg$ \textit{naturality laws for type codes}}]
  \begin{mathparsmall}
    \ebrule{
      \hypo{i,j,k:\SortLvl}
      \hypo{\_:i^+\leq j}
      \hypo{\_:j\leq k}
      \infer3{
        \Lift{j}{k}{\CodeUniv_i^j} = \CodeUniv_i^k : \SortTm\prn{\TpUniv_k}
      }
    }
    \and
    \ebrule{
      \hypo{i,j:\SortLvl}
      \hypo{\_: i\leq j}
      \hypo{M:\SortTm\prn{\TpUniv_i}}
      \hypo{N:\SortTm\prn{\TpEl_i\prn{M}} \to \SortTm\prn{\TpUniv_i}}
      \infer4{
        \Lift{i}{j}\prn{
          \CodePi_i\prn{M,N}
        }
        =
        \CodePi_j\prn[\big] {
          \Lift{i}{j}{M},
          \lambda x\mathpunct{.}
          \Lift{i}{j}{N\prn{x}}
        }
        : \SortTm\prn{\TpUniv_j}
      }
    }
  \end{mathparsmall}
\end{rulebox}

\subsection{A Theory with Injective Decoding and Coercion}

We finally consider the extension \ThyMLTTNatInj{} of \ThyMLTTNat{} by an injectivity law for the universe decoding operation.

\begin{rulebox}[title={\ThyMLTTNatInj{} $\gg$ \textit{injectivity of decoding}}]
  \begin{mathparsmall}
    \ebrule{
      \hypo{i:\SortLvl}
      \hypo{M,N:\SortTm\prn{\TpUniv_i}}
      \hypo{
        \TpEl_i\prn{M} = \TpEl_i\prn{N} : \SortTp
      }
      \infer3{
        M = N : \SortTm\prn{\TpUniv_i}
      }
    }
  \end{mathparsmall}
\end{rulebox}

\begin{corollary}
  The lifting coercions are injective in \ThyMLTTNatInj, \emph{i.e.}\ the following rule is \emph{derivable}:
  \begin{mathparpagebreakable}
    \ebrule{
      \hypo{i,j:\SortLvl}
      \hypo{\_:i\leq j}
      \hypo{M,N:\SortTm\prn{\TpUniv_i}}
      \hypo{
        \Lift{i}{j}M = \Lift{i}{j}N : \SortTm\prn{\TpUniv_j}
      }
      \infer4{
        M = N : \SortTm\prn{\TpUniv_i}
      }
    }
  \end{mathparpagebreakable}
\end{corollary}

\section{``Fuss-free'' Universes via Judgemental Smallness}\label{sec:fuss-free}

We now introduce a second take on the rules of cumulative universes based on a judgemental notion of \emph{smallness}. We will show in this section that our new theory, dubbed \ThyFussFree, is equivalent to the more complex theory \ThyMLTTNatInj{} of natural universe hierarchies with injective decoding. As we have alluded in \zcref{sec:judgemental-smallness}, we shall reformulate the universe hierarchy by replacing each universe $\TpUniv_i$ with its image in the sort of all types. We define the theory \ThyFussFree{} extending $\ThyMLTT{}\cup\ThyLevelPoly{}$ by a new proof-irrelevant form of judgement governing the size of types:

\begin{rulebox}[title={\ThyFussFree{} $\gg$ \textit{a judgemental notion of smallness}}]
  \begin{mathparsmall}
    \ebrule{
      \hypo{A:\SortTp}
      \infer1{\SortSmall_i\prn{A}: \SORT}
    }
    \and
    \ebrule{
      \hypo{A:\SortTp}
      \hypo{p,q:\SortSmall_i\prn{A}}
      \infer2{
        p = q : \SortSmall_i\prn{A}
      }
    }
  \end{mathparsmall}
\end{rulebox}

Now the universes can be defined as a \emph{negative} type connective with formation, introduction, elimination, $\beta$-, and extensionality rules:

\begin{rulebox}[title={\ThyFussFree{} $\gg$ \textit{rules for universes}}]
  \begin{mathparsmall}
    \ebrule{
      \hypo{i:\SortLvl}
      \infer1{\TpUniv_i:\SortTp}
    }
    \and
    \ebrule{
      \hypo{i:\SortLvl}
      \hypo{A:\SortTp}
      \hypo{\_:\SortSmall_i\prn{A}}
      \infer3{
        \Shrink{i} A : \SortTm\prn{\TpUniv_i}
      }
    }
    \and
    \ebrule{
      \hypo{i:\SortLvl}
      \hypo{A:\SortTm\prn{\TpUniv_i}}
      \infer2{\Enlarge{i}{A}:\SortTp}
    }
    \and
    \ebrule{
      \hypo{i:\SortLvl}
      \hypo{A:\SortTm\prn{\TpUniv_i}}
      \infer2{\_: \SortSmall_i\prn[\big]{\Enlarge{i}{A}}}
    }
    \and
    \ebrule{
      \hypo{i:\SortLvl}
      \hypo{A:\SortTp}
      \hypo{\_:\SortSmall_i\prn{A}}
      \infer3{
        \Enlarge{i}\Shrink{i}A = A: \SortTp
      }
    }
    \and
    \ebrule{
      \hypo{A,B:\SortTm\prn{\TpUniv_i}}
      \hypo{\Enlarge{i}{A} = \Enlarge{i}{B}:\SortTp}
      \infer2{A = B: \SortTm\prn{\TpUniv_i}}
    }
  \end{mathparsmall}
\end{rulebox}

To close each universe under type theoretic connectives, we simply add the corresponding rules to the smallness judgement. We emphasise that the entire system is parametric in these rules, in a style similar to the Pure Type Systems of yesteryear.

\begin{rulebox}[title={\ThyFussFree{} $\gg$ \textit{closure properties of smallness}}]
  \begin{mathparsmall}
    \ebrule{
      \hypo{
        \begin{array}{l}
          i:\SortLvl\qquad
          A:\SortTp\qquad
          B:\SortTm\prn{A}\to\SortTp\\
          \_ : \SortSmall_i\prn{A}\qquad
          \_ : \prn{x:\SortTm\prn{A}}\to \SortSmall_i\prn{B\prn{x}}
        \end{array}
      }
      \infer1{
        \_:\SortSmall_i\prn{
          \Pi\prn{A,B}
        }
      }
    }
    \and
    \ebrule{
      \hypo{i,j:\SortLvl}
      \hypo{\_:i^+\leq j}
      \infer2{
        \_:\SortSmall_j\prn{\TpUniv_i}
      }
    }
    \and
    \ebrule{
      \hypo{i:\SortLvl}
      \hypo{A : \SortLvl \rightarrow \SortTm}
      \hypo{\_ : \prn{\alpha:\SortLvl}\to \SortSmall_i\prn{A\prn{\alpha}}}
      \infer3{
        \_:\SortSmall_i\prn{\TpLvlPi\prn{A}}
      }
    }
  \end{mathparsmall}
\end{rulebox}

\begin{example}
    The type $\TpLvlPi\prn{\lambda \alpha\mathpunct{.} \TpUniv_\alpha}$ is not small at any $i:\SortLvl$.
      
\end{example}

\NewDocumentCommand\MOD{m}{\textbf{Mod}\prn{#1}}

\subsection{Semantic Equivalence of \texorpdfstring{\ThyMLTTNatInj}{NatHierInj} and \texorpdfstring{\ThyFussFree}{FussFree}}\label{sec:semantic-equivalence}

It is not difficult to give a bi-interpretation between $\ThyFussFree$ and $\ThyMLTTNatInj$. We can go further, however, and show that their syntactic categories \emph{qua} free categories with representable maps~\citep{uemura-2021-thesis} are equivalent. This implies that the two theories are \emph{semantically} equivalent in the sense that the syntactic translation induces a bi-equivalence between (2,1)-categories of models.

Letting $\mathbb{S}$ and $\mathbb{T}$ be the free categories with representable maps (CwRs) generated by $\ThyFussFree$ and $\ThyMLTTNatInj$ respectively, we shall construct an equivalence of CwRs from $\mathbb{S}$ to $\mathbb{T}$. To define a morphism of CwRs from $F\colon \mathbb{S}\to\mathbb{T}$ is, by the universal property, equivalent to constructing an algebra $\mathsf{M}_F$ for all the operations and equations of $\ThyFussFree$ in terms of those of $\ThyMLTTNatInj$. We allow the algebra $\mathsf{M}_F$ to inherit all constructs that it can, leaving the following for us to define:
\begin{mathpar}
  \textstyle
  \mathsf{M}_F.\SortSmall_i\prn{A} \coloneq
  \sum_{\prn{M:\SortTm\prn{\TpUniv_i}}}
  \TpEl_i\prn{M} = A
  \and
  \mathsf{M}_F.%
  \Shrink{i}\prn{
    A, \prn{M,\_}
  }
  \coloneq M
  \and
  \mathsf{M}_F.\Enlarge{i}\prn{M} \coloneq
  \TpEl_i\prn{M}
\end{mathpar}

The needful equational rules are all satisfied, as are the rules for smallness; the critical case is the proof-irrelevance law for smallness in \ThyFussFree, which follows immediately from the injectivity of decoding in \ThyMLTTNatInj.
Conversely, to construct a morphism of CwRs $G\colon \mathbb{T}\to\mathbb{S}$ is equivalent by the universal property of $\mathbb{T}$ to constructing an algebra $\mathsf{M}_G$ for all the operations and equations of \ThyMLTTNatInj{} in terms of those of \ThyFussFree.
\begin{mathpar}
  \mathsf{M}_G.\TpEl_i\prn{M} \coloneq \Enlarge{i}M
  \and
  \mathsf{M}_G.\Lift{i}{j}{M} \coloneq \Shrink{j}\Enlarge{i}M
  \and
  \mathsf{M}_G.\CodeUniv_i^j \coloneq \Shrink{j}\TpUniv_i
  \and
  \mathsf{M}_G.\CodePi_i\prn{M,N} \coloneq \Shrink{i}\prn{
    \TpPi\prn{
      \Enlarge{i} M,
      \lambda x\mathpunct{.}
      \Enlarge{i}N\prn{x}
    }
  }
  \and
  \mathsf{M}_G.\CodeLvlPi^i\prn{M} \coloneq \Shrink{i}\prn{
    \TpLvlPi\prn{
      \lambda \alpha\mathpunct{.}
      \Enlarge{i} M\prn{\alpha}
    }
  }
\end{mathpar}

\begin{theorem}
  The translations above induce a CwR-equivalence between the syntactic categories of \ThyMLTTNatInj{} and \ThyFussFree, and so their (2,1)-categories of models are (2,1)-equivalent.
\end{theorem}

\begin{proof}
  To construct the unit and counit of the equivalence, we may again use the universal properties of $\mathbb{S}$ and $\mathbb{T}$.
  The only non-trivial case concerns the round-trip translation of smallness:
  \[ 
    \mathsf{M}_{G\circ F}.\SortSmall_i\prn{A} 
    =
    \prn[\big]{
      \textstyle\sum_{\prn{M:\mathsf{M}_G.\SortTm\prn{\mathsf{M}_G.\TpUniv_i}}}
      \mathsf{M}_G.\TpEl_i\prn{M} = A
    }
    =
    \prn[\big]{
      \textstyle\sum_{\prn{M:\SortTm\prn{\TpUniv_i}}}
      \Enlarge{i}M = A
    }
    \cong
    \SortSmall_i\prn{A}
  \] 

  But for the above, $G\circ F$ would be the identity functor; instead, it is a natural isomorphism.
\end{proof}

\section{Extension: Fuss-Free Inductive Datatypes}\label{sec:fuss-free-datatypes}

We now show that the fuss-free presentation admits a particularly simple account of inductive datatypes, including large ones. Our starting point is the PhD thesis of \citet{dagand:2013}, who studied datatype descriptions and their elaboration.

Dagand considered a non-cumulative hierarchy of universes of datatype descriptions; we will show in this section how the fuss-free methodology can be used to easily develop a cumulative version of Dagand's hierarchy without sacrificing the algebraic character of the system.

\subsection{Judgemental Datatype Descriptions}

Rather than first introducing stratified types of datatype descriptions at each universe level, we start with a purely judgemental notion of description.

\begin{rulebox}[title={\ThyFussFreeData{} $\gg$ \textit{judgemental datatype descriptions}}]
  \begin{mathparsmall}
    \ebrule{
      \infer0{\SortDesc: \SORT}
    }
    \and
    \ebrule{
      \infer0{
        \DescVar, \DescUnit:\SortDesc
      }
    }
    \and
    \ebrule{
      \hypo{D:\SortDesc}
      \hypo{E:\SortDesc}
      \infer2{
        D \DescTensor E : \SortDesc
      }
    }
    \and
    \ebrule{
      \hypo{A:\SortTp}
      \hypo{D:\SortTm\prn{A}\to \SortDesc}
      \infer2{
        \bigsqcup\DescSigil\prn{A,D},
        \bigsqcap\DescSigil\prn{A,D}:\SortDesc
      }
    }
  \end{mathparsmall}
\end{rulebox}

Next we introduce the \emph{extension} $\bbrk{D}:\SortTp\to\SortTp$ of a description $D:\SortDesc$.
In the following, we will freely use unit and dependent sum types, whose treatment is analogous to dependent products. We also assume finite enumeration types (see \citet{dagand:2013}), which are naturally small at all levels.

\begin{rulebox}[title={\ThyFussFreeData{} $\gg$ \textit{the extension operator for descriptions}}]
  \begin{mathparsmall}
    \ebrule{
      \hypo{D:\SortDesc}
      \hypo{X:\SortTp}
      \infer2{\bbrk{D}X:\SortTp}
    }
    \and
    \ebrule{
      \hypo{X:\SortTp}
      \infer1{\bbrk{\DescVar}X = X : \SortTp}
    }
    \and
    \ebrule{
      \hypo{X:\SortTp}
      \infer1{\bbrk{\DescUnit}X = \TpUnit : \SortTp}
    }
    \and
    \ebrule{
      \hypo{D,E:\SortDesc}
      \hypo{X:\SortTp}
      \infer2{
        \bbrk{D \DescTensor E}X = \bbrk{D}X \times \bbrk{E}X:\SortTp
      }
    }
    \and
    \ebrule{
      \hypo{A:\SortTp}
      \hypo{D:\SortTm\prn{A}\to \SortDesc}
      \hypo{X:\SortTp}
      \infer3{
        \bbrk[\big]{
          \bigsqcup\DescSigil\prn{A,D}
        }X =
        \Sigma\prn[\big]{A, \lambda x\mathpunct{.} \bbrk{D\prn{x}}X}
      }
    }
    \and
    \ebrule{
      \hypo{A:\SortTp}
      \hypo{D:\SortTm\prn{A}\to \SortDesc}
      \hypo{X:\SortTp}
      \infer3{
        \bbrk[\big]{
          \bigsqcap\DescSigil\prn{A,D}
        }X =
        \Pi\prn[\big]{A, \lambda x\mathpunct{.} \bbrk{D\prn{x}}X}
      }
    }
  \end{mathparsmall}
\end{rulebox}

Following \citet{dagand:2013}, we will also have a functorial action for each $\bbrk{D}$ as well as a family lifting operation $\square$.
We omit their equational laws for lack of space, but they can be found in \emph{op.\ cit.}

\begin{rulebox}[title={\ThyFussFreeData{} $\gg$ \textit{functorial action of extension; family lifting}}]
  \begin{mathparsmall}
    \ebrule{
      \hypo{D:\SortDesc}
      \hypo{X,Y:\SortTp}
      \hypo{f : \SortTm\prn{X}\to \SortTm\prn{Y}}
      \hypo{M : \SortTm\prn[\big]{\bbrk{D}X}}
      \infer4{
        \fmap\brc{D,X,Y}\prn{f,M} : \SortTm\prn[\big]{\bbrk{D}Y}
      }
    }
    \and
    \ebrule{
      \hypo{D:\SortDesc}
      \hypo{X:\SortTp}
      \hypo{C:\SortTm\prn{X}\to\SortTp}
      \hypo{M:\SortTm\prn[\big]{\bbrk{D}X}}
      \infer4{\square\brc{D,X}\prn{C, M}:\SortTp}
    }
    \\
    \textit{(equational laws omitted for space)}
  \end{mathparsmall}
\end{rulebox}

The $\square$ operation governs to the structural induction hypotheses of the datatype: if we see the family $C$ as a property on $X$, then $\square\brc{D,X}\prn{C, M}$ states that $C$ should hold on all the subnodes $x:X$ of $M$ \emph{i.e.} everywhere in the datatype $D$.
With this in hand, we can specify the rules for the datatype constructor associated to a datatype description.

\begin{rulebox}[title={\ThyFussFreeData{} $\gg$ \textit{the datatype constructor}}]
  \begin{mathparsmall}
    \ebrule{
      \hypo{D:\SortDesc}
      \infer1{
        \TpData\prn{D} : \SortTp
      }
    }
    \and
    \ebrule{
      \hypo{D:\SortDesc}
      \hypo{M:\SortTm\prn[\big]{\bbrk{D}\prn{\TpData\prn{D}}}}
      \infer2{
        \TmMake\brc{D}\prn{M} : \SortTm\prn{\TpData\prn{D}}
      }
    }
    \and
    \ebrule{
      \hypo{
        \begin{array}{l}
          D:\SortDesc \qquad
          C:\SortTm\prn{\TpData\prn{D}}\to \SortTp
          \\
          M:c
            \prn[\big]{x:\SortTm\prn[\big]{\bbrk{D}\prn{\TpData\prn{D}}}}
            \to
            \SortTm\prn[\big]{\square\brc{D,\TpData\prn{D}}\prn{C,x}}
            \to \SortTm\prn{C\prn{\TmMake\brc{D}\prn{x}}}
          \\
          N : \SortTm\prn[\big]{\TpData\prn{D}}
        \end{array}
      }
      \infer1{
        \TmElim\brc{D}\prn{C,M,N} : \SortTm\prn{
          C\prn{N}
        }
      }
    }
    \\
    \textit{(equational laws omitted for space)}
  \end{mathparsmall}
\end{rulebox}

\begin{example}[Lists]
    We can define a description for the type of lists as follows, using the syntax of finite enumeration types as described by \citet{dagand:2013}:
    \[ 
    \TpList\DescSigil(A) \coloneq
    \bigsqcup\nolimits\DescSigil
    \begin{pmatrix*}[l]
     \brc{`\mathtt{nil}, `\mathtt{cons}},
      \\
      \begin{cases}
        `\mathtt{nil} &\mapsto \DescUnit\\
        `\mathtt{cons} &\mapsto \bigsqcup\DescSigil\prn{A, \lambda X\mathpunct{.} \DescVar \DescTensor \DescUnit}
      \end{cases}
    \end{pmatrix*}
    \] 

    Then, by unfolding the definition of its extension (\emph{i.e.} interpretation) $\bbrk{\TpList\DescSigil(A)}$, we recover the well-known signature functor $X \mapsto 1 + A \times X$.
\end{example}

\begin{example}[W-types]
    We can define a description for W-types (well-founded trees) as follows, given a type $A:\SortTp$ and  a family $B:\SortTm\prn{A} \rightarrow \SortTp$, setting
    $
      \mathrm{W}\DescSigil\prn{A,B} \coloneq  
      \bigsqcup\nolimits\DescSigil
      \prn[\big]{A,\lambda x\mathpunct{.} \bigsqcap\nolimits\DescSigil \prn{B\prn{x},\lambda \_ \mathpunct{.}\DescVar}}
      \text{.}
    $
    We recover the expected signature endofunctor $\bbrk{\mathrm{W}\DescSigil(A,B)}X = \Sigma\prn{A, \lambda x\mathpunct{.}\Pi\prn{B(x),\lambda \_\mathpunct{.}X}}$.
\end{example}

We could similarly add indexed descriptions to our framework to cover general inductive families. We note that in the presence of identity types, data type descriptions have the same expressivity as W-types via Hugunin's Trick~\citep{hugunin:2020}.

\subsection{The Fuss-Free Approach to Datatype Smallness}

When is a datatype $\TpData\prn{D}$ small? We could easily define a separate smallness judgement for \emph{descriptions}, but it so happens that we may soundly express the smallness of a datatype description in terms of the smallness of its \emph{extension}. In particular, we introduce the proof-irrelevant sort $\SortSmallD^i\prn{D}$ to \ThyFussFreeData{} as a definitional extension:
\[
  \SortSmallD^i\prn{D} \triangleq
  \prn{x : \SortTm\prn{\TpUniv_i}}\to
  \SortSmall_i\prn[\big]{
    \bbrk{D}\prn{\TpEl_i\prn{x}}
  }
\]

With this in hand, we may introduce a rule to make datatypes small when appropriate:

\begin{rulebox}[title={\ThyFussFreeData{} $\gg$ \textit{smallness for datatypes}}]
  \begin{mathparsmall}
    \ebrule{
      \hypo{i:\SortLvl}
      \hypo{D:\SortDesc}
      \hypo{\_ : \SortSmallD^i\prn{D}}
      \infer3{
        \_:\SortSmall_i\prn{\TpData\prn{D}}
      }
    }
  \end{mathparsmall}
\end{rulebox}

\begin{example}[Small Lists]
    For the datatype description of lists, the extension $\bbrk{\TpList\DescSigil\prn{A}}$ instantiated at a decoded term $\TpEl_i\prn{x}$ is the type $\TpUnit + A \times \TpEl_i\prn{x}$. This type satisfies the smallness predicate only if $A$ is small itself, such that we may derive the desired rule:
  \begin{mathparsmall}
    \ebrule{
      \hypo{
        \_:\SortSmall_i\prn{A}
      }
      \infer1{
        \_:\SortSmall_i\prn{\TpList\prn{A}}
      }
    }
  \end{mathparsmall}
  
  For example, if we have $A:\TpUniv_i$  then the type of lists $\TpList\prn{\Enlarge{i}{A}}$ is itself small at size $i$. Although we could define a polymorphic code $\mathsf{list}\colon\TpLvlPi\prn{\lambda \alpha\mathpunct{.}\TpUniv_\alpha\to\TpUniv_\alpha}$ satisfying appropriate coherence equations, in some sense the point of the fuss-free methodology is that we do not need to do this: many constructs that require parameterising in universe levels in other systems can be expressed directly as types and then shrunk down to an appropriate universe at will.
\end{example}

\begin{example}[Small W-types]
    For the datatype description of W-types, the extension $\bbrk{\mathrm{W}\DescSigil(A,B)}$ instantiated at a decoded term $\TpEl_i\prn{x}$ is the type $\Sigma\prn{A, \lambda x\mathpunct{.}\Pi\prn{B(x),\lambda \_\mathpunct{.}\TpEl_i\prn{x}}}$.
    This type satisfies the smallness predicate only if $A$ and $B$ are small, such that we derive the desired rule:
  \begin{mathparsmall}
    \ebrule{
      \hypo{
        \_:\SortSmall_i\prn{A}
      }
      \hypo{
          \_: (x: \SortTm\prn{A}) \rightarrow \SortSmall_i\prn{B(x)}
      }
      \infer2{
        \_:\SortSmall_i\prn{\mathrm{W}\prn{A,B}}
      }
    }
  \end{mathparsmall}
\end{example}

\subsection{Universes of Datatype Descriptions}

So far, we have treated datatype descriptions only as a tool to unify the specification of inductive datatypes. Their original motivation, however, is to support \emph{metaprogramming} of two forms:
\begin{enumerate}
  \item \textbf{Generative:} allowing datatype descriptions to be computed at ``runtime''.
  \item \textbf{Introspective:} defining generic programs on arbitrary datatypes by recursion on the description (\emph{e.g.}\ pretty-printing).
\end{enumerate}

We will see how to support both in a staged manner. For generative metaprogramming, we need only a predicative hierarchy of \emph{open} or \emph{Russell-style} universes of datatype descriptions. We summarise the rules for these below:

\begin{rulebox}[title={\ThyFussFreeData{} $\gg$ \textit{universe of descriptions \`a la Russell}}]
  \begin{mathparsmall}
    \ebrule{
      \hypo{i:\SortLvl}
      \infer1{
        \TpDesc_i:\SortTp
      }
    }
    \and
    \ebrule{
      \hypo{i,j:\SortLvl}
      \hypo{\_:i^+\leq j}
      \infer2{
        \_ : \SortSmall_j\prn{\TpDesc_i}
      }
    }
    \and
    \ebrule{
      \hypo{i:\SortLvl}
      \hypo{D:\SortDesc}
      \hypo{\_:\SortSmallD^i\prn{D}}
      \infer3{
        \ShrinkD{i}D : \SortTm\prn{\TpDesc_i}
      }
    }
    \and
    \ebrule{
      \hypo{i:\SortLvl}
      \hypo{M:\SortTm\prn{\TpDesc_i}}
      \infer2{
        \EnlargeD{i}M : \SortDesc
      }
    }
    \and
    \ebrule{
      \hypo{i:\SortLvl}
      \hypo{M:\SortTm\prn{\TpDesc_i}}
      \infer2{
        \_:\SortSmallD^i\prn{\EnlargeD{i}M}
      }
    }
    \and
    \ebrule{
      \hypo{i:\SortLvl}
      \hypo{D:\SortDesc}
      \hypo{\_ : \SortSmallD^i\prn{D}}
      \infer3{
        \EnlargeD{i}\ShrinkD{i}D = D : \SortDesc
      }
    }
    \and
    \ebrule{
      \hypo{M, N : \SortTm\prn{\TpDesc_i}}
      \hypo{\EnlargeD{i}M = \EnlargeD{i}N : \SortDesc}
      \infer2{M = N : \SortTm\prn{\TpDesc_i}}
    }
  \end{mathparsmall}
\end{rulebox}

\paragraph{Bootstrapping Closed Universes of Descriptions}

For introspective metaprogramming, where we may define things by recursion on a description's structure, we would need \emph{closed} or \emph{Tarski-style} universes of datatype descriptions. Taking a leaf from the \emph{gentle art of levitation} of \citet{levitation}, we note that these are already definable! 
In particular, we may define descriptions $\mathsf{DescD}_i:\SortDesc$ at each level $i:\SortLvl$ that encode the structure of small datatype descriptions:

\[ 
  \mathsf{DescD}_i \coloneq 
  \bigsqcup\nolimits\DescSigil
  \begin{pmatrix*}[l]
   \brc{`\mathtt{var}, `\mathtt{unit}, `\mathtt{tensor}, `\mathtt{prod}, `\mathtt{coprod}},
    \\
    \begin{cases}
      `\mathtt{var},`\mathtt{unit} &\mapsto \DescUnit\\
      `\mathtt{tensor} &\mapsto \DescVar\DescTensor\DescVar\\
      `\mathtt{prod},`\mathtt{coprod} &\mapsto \bigsqcup\DescSigil\prn[\big]{\TpUniv_i, \lambda X\mathpunct{.} \bigsqcap\DescSigil\prn{\Enlarge{i}{X},\lambda \_\mathpunct{.}\DescVar}}
    \end{cases}
  \end{pmatrix*}
\] 

Then we may define the \emph{universe} of datatype descriptions at level $i$ to be the datatype $\TpData\prn{\mathsf{DescD}_i}$. We conclude with a few quick observations:
\begin{enumerate}
\item The description $\mathsf{DescD}_i$ is easily seen to be $i^+$-small, and therefore so is the type $\TpData\prn{\mathsf{DescD}_i}$.%
\item With the eliminator for $\TpData\prn{\mathsf{DescD}_i}$, we may define a \emph{decoding} function $\ceils{-}_i$ from $\TpData\prn{\mathsf{DescD}_i}$ to $\TpDesc_i$ that turns encoded descriptions into actual descriptions.
\item Hence for any $d:\TpData\prn{\mathsf{DescD}_i}$ we have an $i$-small description $\EnlargeD{i}\ceils{d}_i : \SortDesc$ and hence an $i$-small datatype $\TpData\prn[\big]{\EnlargeD{i}\ceils{d}_i}$.
\end{enumerate}

In this way, we treat both generative and introspective metaprogramming in the fuss-free presentation of datatype descriptions and their universes.

\section{Metatheory: Injectivity of Decoding and Coercion in \texorpdfstring{\ThyMLTTNat}{NatHierInj}}\label{sec:normalisation}

In this section, we will prove that in the \emph{initial} or \emph{syntactic} model of \ThyMLTTNat, the universe decoding operation $A:\TpUniv_i\vdash \TpEl_i\prn{A}:\SortTp$ is injective; this does not mean that the injectivity law from \ThyMLTTNatInj{} is derivable in \ThyMLTTNat{} but rather that it is \emph{admissible}. Admissible rules are special properties of syntax that need not be preserved under homomorphisms of models.

Unlike simpler metatheorems (\emph{e.g.}\ admissibility of substitution), normalisation for full Martin-L\"of type theory probably cannot be proved directly by induction on typing derivations. Instead, a variation on Tait's method of computability has proved the most scalable for proving normalisation results.
In the context of second-order generalised algebraic theories, it is particularly convenient to make use of Bocquet's \emph{relative induction principle}~\citep{BocquetPhD} which re-casts Tait's method in semantic terms for an \emph{arbitrary} SOGAT. Using this induction principle, normalisation for \ThyMLTTNat{} can be proved by glueing along the global sections functor of its initial model \emph{relativised} in a certain way into the category of presheaves on its category of renamings.

\subsection{Higher-Order and First-Order Models}

Every SOGAT can be read as the signature of a dependent record type involving dependent function types in its fields. A \emph{higher-order model} of such a theory is simply defined to be an element of this dependent record type. This notion of higher-order model can, of course, be applied in any category where dependent type theory is interpreted, as we have done in \zcref{sec:semantic-equivalence}.

Higher-order models are oblivious to the representability structure of the SOGAT, and therefore do not treat \emph{contexts}. An ordinary model of dependent type theory (\emph{e.g.}\ a CwF) starts with a category of contexts $\mathbb{C}$, and then the rest of the type theory is interpreted in presheaves on $\mathbb{C}$ in such a way that contexts become representable. Models that incorporate contexts are referred to by \citet{BocquetPhD} as \emph{first-order models}:
\begin{quote}
  A \emph{first-order model} of $\mathbb{T}$ is a small category $\mathbb{C}_{\mathcal{M}}$ with a terminal object together with a higher-order model $\mathcal{M}$ in $\Psh\prn{\mathbb{C}_{\mathcal{M}}}$ such that every representable family of sorts $x:B\vdash E\prn{x}:\REPSORT$ is interpreted by a representable family in $\Psh\prn{\mathbb{C}_{\mathcal{M}}}/\bbrk{B}_{\mathcal{M}}$.
\end{quote}

As described above, the notion of first-order model of $\mathbb{T}$ is essentially algebraic and therefore can be interpreted in any category with finite limits, as a model of a finite limit theory that we might call $\mathbb{T}^{\mathsf{fo}}$; this gives a notion of \emph{internal} first-order model, whose category of contexts is then an internal category. The notion of internal first-order model plays an important role in Bocquet's \emph{relative induction principle}, which we shall use to prove the normalisation theorem for $\ThyMLTTNat$.

\subsection{Abstract Renamings and the Relative Induction Principle}

\NewDocumentCommand\ThyCwF{}{\mathbb{T}_0}
\NewDocumentCommand\REN{}{\mathsf{Ren}}

Semantic normalisation theorems always rely on distinguishing a class of substitutions called \emph{renamings} under which normal forms shall be stable from the outset, as explained by \citet{fiore:2002}. For simple type theory, it is very easy to write down explicitly what the renamings are: if a context is a finite set labelled by types, then the renamings are context morphisms (simultaneous substitutions) that can be formed using only variables. In the presence of dependent sorts, it is comparatively more difficult to explicitly write down a reasonable notion of renaming; there is, however, a universal property that guides one to the correct definition:

\begin{definition}[Level-polymorphic renaming algebra]\label{def:renaming-algebra}
  Let $\ThyCwF$ be the SOGAT consisting of a representable sort of levels $\SortLvl : \REPSORT$ and a family of representable sorts $A:\SortTp\vdash \SortTm\prn{A}:\REPSORT$,\footnote{In other words, the SOGAT whose first-order models are level-indexed CwFs.} and let $\mathbb{T}$ be a SOGAT extending $\mathbb{T}_0$. A \emph{level-polymorphic renaming algebra} over a first-order model $\ModS$ of $\mathbb{T}$ is given by a first-order model $\mathcal{R}$ of $\ThyCwF$ equipped with a homomorphism $\rho\colon\mathcal{R}\to \ModS_{\vert\mathbb{T}_0}$ of $\ThyCwF$-models whose component at $\SortTp$ is an isomorphism $\rho_{\SortTp} \colon \mathcal{R}.\SortTp \cong \rho^*\ModS.\SortTp \in \Psh\prn{\mathbb{C}_{\mathcal{R}}}$. We shall abuse notation by leaving this isomorphism implicit.
\end{definition}

\subsubsection{The Initial Level-Polymorphic Renaming Algebra}
In the context of Definition~\ref{def:renaming-algebra}, we shall write $\rho\colon \REN_{\ModS}\to \ModS_{\vert \ThyCwF}$  for the \emph{initial} level-polymorphic renaming algebra over the model $\ModS$. For brevity, we shall write $\mathbb{R}_{\ModS} \coloneq \mathbb{C}_{\REN_{\ModS}}$ for the category of contexts of the initial level-polymorphic renaming algebra; then $\rho\colon \mathbb{R}_{\ModS}\to\mathbb{C}_{\ModS}$ is the ``inclusion'' of contexts and renamings into contexts and substitutions. For each $A:\REN_{\ModS}.\SortTp$, we shall write
$
  \mathsf{Var}_{\ModS}\prn{A} \coloneq
  \REN_{\ModS}.\SortTm\prn{A}
$
and
$
  \LvlVar_{\ModS} \coloneq
  \REN_{\ModS}.\SortLvl
$
for the presheaves of terms and levels in the level-polymorphic renaming algebra, which we shall refer to evocatively as ``variables'' and ``level variables''.

\subsubsection{Variable Parameterisation}

We shall now make a copy of the first-order model $\ModS$ inside of $\Psh\prn{\mathbb{R}_{\ModS}}$ in which the ambient variables from $\Var_{\ModS}\prn{A}$ and $\LvlVar_{\ModS}$ are incorporated as \emph{closed} terms. We call this the \emph{variable parameterisation} of $\ModS$, which we shall write $\VarParam{\ModS}$. To define $\VarParam{\ModS}$ as an internal first-order $\mathbb{T}$-model, we construct a finitely continuous functor $\VarParam{\ModS}\colon \mathbb{T}^{\mathsf{fo}}\to\Psh\prn{\mathbb{R}_{\ModS}}$ whose component
$\VarParam{\ModS}\prn{X}_{\Psi}$ for $X\in\mathbb{T}^{\mathsf{fo}}$ and $\Psi\in\mathbb{R}_{\ModS}$ is defined to be the corresponding component $\ModS\brk{\psi\colon \mathbf{1}\to\rho\prn{\Psi}}\prn{X}$ of the free extension of $\ModS$ by a formal parameter $\psi\colon\mathbf{1}\to\rho\prn{\Psi}$.

Nearly by definition, we have an isomorphism $\VarParam{\ModS}.\SortTp\prn{\mathbf{1}} \cong \REN_{\ModS}.\SortTp$ in $\Psh\prn{\mathbb{R}_{\ModS}}$ which we shall leave implicit.
The variable parameterisation induces, internally, the following constructors:
\begin{mathparsmall}
  \ebrule{
    \hypo{A:\REN_{\ModS}.\SortTp}
    \hypo{x:\mathsf{Var}_{\ModS}\prn{A}}
    \infer2{\mathsf{var}\prn{A,x} : \VarParam{\ModS}.\SortTm\prn{\mathbf{1},A}}
  }
  \and
  \ebrule{
    \hypo{\alpha:\LvlVar_{\ModS}}
    \infer1{\lvlVar\prn{\alpha} : \VarParam{\ModS}.\SortLvl\prn{\mathbf{1}}}
  }
\end{mathparsmall}

The following lemma expresses the precise sense in which $\VarParam{\ModS}$ is ``parameterised'' in variables: binders are equivalently represented as open terms and as families of closed terms.

\begin{lemma}\label{lemma:substitution-is-iso}
  The following substitution functions induced by $\var$ and $\lvlVar$ are all isomorphisms:
  \begin{align*}
    \mathsf{subst}_{\SortTm,\SortTp} &: \VarParam{\ModS}.\SortTp\prn{\mathbf{1}.A} \to \prn[\big]{\Var_\ModS\prn{A}\to\VarParam{\ModS}.\SortTp\prn{\mathbf{1}}}
    \\ 
    \mathsf{subst}_{\SortLvl,\SortTp} &: \VarParam{\ModS}.\SortTp\prn{\mathbf{1}.\SortLvl} \to \prn[\big]{\LvlVar_\ModS\to\VarParam{\ModS}.\SortTp\prn{\mathbf{1}}}
    \\
    \mathsf{subst}_{\SortTm,\SortTm} &: \VarParam{\ModS}.\SortTm\prn{\mathbf{1}.A,B} \to \textstyle\prod_{x:\Var_\ModS\prn{A}}\VarParam{\ModS}.\SortTm\prn{\mathbf{1},\mathsf{subst}_{\SortTm,\SortTp}\prn{B}\prn{x}}
    \\ 
    \mathsf{subst}_{\SortLvl,\SortTm} &: \VarParam{\ModS}.\SortTm\prn{\mathbf{1}.\SortLvl,B}\to \textstyle\prod_{x:\LvlVar_\ModS}\VarParam{\ModS}.\SortTm\prn{\mathbf{1},\mathsf{subst}_{\SortLvl,\SortTp}\prn{B}\prn{x}}
  \end{align*}
\end{lemma}

\subsection{Normalisation of \texorpdfstring{\ThyMLTTNat}{NatHier}}

\subsubsection{Inductive Specification of Normal and Neutral Forms}
Letting $\ModS$ be the initial first-order model of \ThyMLTTNat{} and $\rho\colon\REN_\ModS\to \ModS$ be its initial renaming algebra over $\ModS$. In this section, we will work in $\Psh\prn{\mathbb{R}_{\ModS}}$ within the scope of the variable parameterisation $\VarParam{\ModS}$.
The first step toward a normalisation result is to specify what the normal (and neutral) forms ought to be; because normals are \emph{a priori} stable only under renamings, it is natural to phrase the normals and neutrals as a mutual indexed inductive definition internal to $\Psh\prn{\mathbb{R}_{\ModS}}$:
\begin{mathpar}
  \fbox{
    \ebrule{
      \hypo{A: \VarParam{\ModS}.\SortTp\prn{\mathbf{1}}}
      \infer1{\IsNfTp{A}:\Set}
    }
  }
  \and
  \fbox{
    \ebrule{
      \hypo{A: \VarParam{\ModS}.\SortTp\prn{\mathbf{1}}}
      \hypo{M : \VarParam{\ModS}.\SortTm\prn{\mathbf{1},A}}
      \infer2{\IsNfTm{A}{M}, \IsNeTm{A}{M}:\Set}
    }
  }
\end{mathpar}

In other words, for each syntactic type $A$ we shall define a type $\IsNfTp{A}$ of \emph{normal forms} of that type, and for each syntactic term $M$ of type $A$ we shall define types $\IsNfTm{A}{M}$ and $\IsNeTm{A}{M}$ of normal and neutral forms respectively of $M$. The goal of the normalisation theorem is to show that each $\IsNfTp{A}$ and $\IsNfTm{A}{M}$ are always \emph{singletons}, so that there exists exactly one normal form for every syntactic construct.
The inductive definition of neutral and normal form witnesses is summarised in \zcref{fig:normals-and-neutrals}, and is mostly standard (see \emph{e.g.}\ Sterling~\cite[\S5.2]{sterling-2021-thesis}).

\begin{figure}
  \begin{mathpartiny}
    \fbox{
      \ebrule{
        \hypo{A: \SortTp\prn{\mathbf{1}}}
        \infer1{\IsNfTp{A}:\Set}
      }
    }
    \and
    \fbox{
      \ebrule{
        \hypo{A: \SortTp\prn{\mathbf{1}}}
        \hypo{M : \SortTm\prn{\mathbf{1},A}}
        \infer2{\IsNfTm{A}{M}, \IsNeTm{A}{M}:\Set}
      }
    }
    \\
    \ebrule[nftp-pi, nftp-sg]{
      \hypo{\IsNfTp{A}}
      \hypo{\prn{x:\Var\prn{A}}\to \IsNfTp[\big]{B\brk{\mathsf{var}\prn{x}}}}
      \infer2{
        \IsNfTp{\Pi\prn{A,B}},
        \IsNfTp{\Sigma\prn{A,B}}
      }
    }
    \and
    \ebrule[nftp-lvl-pi]{
      \hypo{(i: \LvlVar) \rightarrow \IsNfTp[\big]{A\brk{\lvlVar\prn{i}}}}
      \infer1{
        \IsNfTp{\TpLvlPi\prn{A}}
      }
    }
    \and
    \ebrule[nftp-univ$_i$]{
      \infer0{\IsNfTp{\TpUniv_i}}
    }
    \and
    \ebrule[nftp-el]{
      \hypo{
        \IsNeTm{\TpUniv_i}{M}
      }
      \infer1{
        \IsNfTp{\TpEl_i\prn{M}}
      }
    }
    \and
    \ebrule[nf-lam]{
      \hypo{\prn{x:\Var\prn{A}}\to \IsNfTm[\big]{B\brk{\var\prn{x}}}{M\brk{\var\prn{x}}}}
      \infer1{
        \IsNfTm{\Pi\prn{A,B}}{\TmLam\brc{A,B}\prn{M}}
      }
    }
    \and
    \ebrule[nf-lvl-lam]{
      \hypo{(i: \LvlVar) \rightarrow \IsNfTm[\big]{A\brk{\mathsf{lvl}\prn{i}}}{u\brk{\lvlVar\prn{i}}}}
      \infer1{
        \IsNfTm{\TpLvlPi\prn{A}}{\TmLvlLam\brc{A}\prn{u}},
      }
    }
    \and
    \ebrule[nf-pair]{
      \hypo{\IsNfTm{A}{M}}
      \hypo{\IsNfTm[\big]{B\brk{M}}{N}}
      \infer2{
        \IsNfTm{\Sigma\prn{A,B}}{\TmPair\brc{A,B}\prn{M,N}}
      }
    }
    \and
    \ebrule[nf-ne-shift]{
      \hypo{\IsNeTm{\TpUniv_i}{M}}
      \hypo{\IsNeTm{\TpEl_i\prn{M}}{N}}
      \infer2{
        \IsNfTm{\TpEl_i\prn{M}}{N}
      }
    }
    \and
    \ebrule[nf-lift\ensuremath{_j}]{
      \hypo{\IsNeTm{\TpUniv_i}{M}}
      \infer1{
        \IsNfTm[\big]{\TpUniv_j}{\Lift{i}{j}{M}}
      }
    }
    \and
    \ebrule[nf-code-univ\ensuremath{_j}]{
      \infer0{
        \IsNfTm{\TpUniv_j}{\CodeUniv_i^j}
      }
    }
    \and
    \ebrule[nf-code-pi, nf-code-sg]{
      \hypo{\IsNfTm{\TpUniv_i}{M}}
      \hypo{\prn{x:\Var\prn{\TpEl_i\prn{M}}}\to \IsNfTm[\big]{\TpUniv_i}{N\brk{\var\prn{x}}}}
      \infer2{
        \IsNfTm{\TpUniv_i}{\CodePi_i\prn{M,N}},
        \IsNfTm{\TpUniv_i}{\CodeSg_i\prn{M,N}}
      }
    }
    \and
    \ebrule[nf-code-lvl-pi]{
      \hypo{\prn{x:\LvlVar}\to \IsNfTm[\big]{\TpUniv_i}{M\brk{\lvlVar\prn{x}}}}
      \infer1{
        \IsNfTm{\TpUniv_i}{\CodeLvlPi^i\prn{M}}
      }
    }
    \and
    \ebrule[ne-var]{
      \infer0{\IsNeTm{A}{\var\prn{A,x}}}
    }
    \and
    \ebrule[ne-app]{
      \hypo{\IsNeTm{\Pi\prn{A,B}}{M}}
      \hypo{\IsNfTm{A}{N}}
      \infer2{
        \IsNeTm[\big]{B\brk{N}}{\TmApp\brc{A,B}\prn{M,N}}
      }
    }
    \and
    \ebrule[ne-lvl-app]{
      \hypo{\IsNeTm{\TpLvlPi\prn{A}}{M}}
      \infer1{
        \IsNeTm[\big]{A\brk{\mathsf{lvl}\prn{i}}}{\TmLvlApp\brc{A}\prn{M,i}}
      }
    }
    \and
    \ebrule[ne-fst]{
      \hypo{\IsNeTm{\Sigma\prn{A,B}}{M}}
      \infer1{\IsNeTm{A}{\TmFst\brc{A,B}\prn{M}}}
    }
    \and
    \ebrule[ne-snd]{
      \hypo{\IsNeTm{\Sigma\prn{A,B}}{M}}
      \infer1{
        \IsNeTm[\big]{B\brk{\TmFst\brc{A,B}\prn{M}}}{
          \TmSnd\brc{A,B}\prn{M}
        }
      }
    }
  \end{mathpartiny}
  \caption{Indexed inductive definition of normal and neutral forms in $\Psh\prn{\mathsf{Ren}_\mathcal{S}}.$ We assume all terms appearing as indices are well-typed, leaving many evident premises implicit.}
  \label{fig:normals-and-neutrals}
\end{figure}

\subsubsection{The Normalisation Nodel}

The next step is to construct a \emph{displayed} model of \ThyMLTTNat{} over variable parameterisation in $\Psh\prn{\mathbb{R}_{\ModS}}$ that equips each type with a (proof-relevant) Tait computability structure and quotation/unquotation yoga~\citep{tait:1967,sterling-2021-thesis,BocquetPhD}. We do not describe this construction here for lack of space, but the key details are included in our supplementary appendices.

\subsection{The Normalisation Theorem; Injectivity of Type Constructors \& Decoding}

A level-polymorphic adaptation of Bocquet's relative induction principle applied to the normalisation displayed model yields the normalisation result.

\begin{theorem}[Internal normalisation of \ThyMLTTNat]\label{thm:normalisation}
  Every closed type and closed term of $\VarParam{\ModS}$ has a unique normal form in the sense that both $\IsNfTp{A}$ and $\IsNfTm{A}{M}$ are contractible for all $A:\VarParam{\ModS}.\SortTp\prn{\mathbf{1}}$  and $M:\VarParam{\ModS}.\SortTm\prn{\mathbf{1},A}$.
\end{theorem}

The result above is internal to $\Psh\prn{\mathbb{R}_{\ModS}}$, where anything said of closed terms ends up being effective (externally) for all open terms. \zcref{thm:normalisation} therefore implies the standard \emph{external} normalisation theorem for open terms, answering a question left open by \citet{bcde-2022}.
We can deduce the standard corollaries of normalisation, including all the standard inversion lemmas:

\begin{corollary}[Injectivity of $\TpPi,\TpLvlPi,\TpSg$]\label{cor:injectivity-of-type-constructors}
  The $\TpPi$, $\TpLvlPi$, and $\TpSg$ constructors are injective in $\VarParam{\ModS}$. Given
  $A,A':\VarParam{\ModS}.\SortTp\prn{\mathbf{1}}$, $B:\VarParam{\ModS}.\SortTp\prn{\mathbf{1}.A}$, and $B':\VarParam{\ModS}.\SortTp\prn{\mathbf{1}.A'}$ such that $\TpPi\prn{A,B} = \TpPi\prn{A',B'}$ or $\TpSg\prn{A,B}=\TpSg\prn{A',B'}$, then we have $A=A'$ and $B=B'$. Similarly, given $B,B':\VarParam{\ModS}.\SortTp\prn{\mathbf{1}.\SortLvl}$ such that $\TpLvlPi\prn{B} = \TpLvlPi\prn{B'}$, we have $B=B'$.
  Hence, externally, the $\TpPi$-, $\TpLvlPi$-, and $\TpSg$-type constructors are injective in all contexts.
\end{corollary}

\begin{corollary}[No-confusion]\label{cor:no-confusion}
  The constructors $\TpPi$, $\TpLvlPi$, $\TpSg$, and $\TpUniv$ are all disjoint.
\end{corollary}

Using all of the above, we finally deduce the injectivity of the decoding family.

\begin{restatable}[Injectivity of decoding]{theorem}{ThmInjDecoding}\label{thm:inj-decoding}
  For any $M,N : \VarParam{\ModS}.\SortTm\prn{\mathbf{1},\TpUniv_i}$ such that $\TpEl_i\prn{M} = \TpEl_i\prn{N}$, we have $M = N$. Hence, externally, the decoding family of\ \ThyMLTTNat{} is injective in all contexts.
\end{restatable}

\begin{corollary}
  The following rule holds in the initial model of\ \ThyMLTTNat{}:
  \begin{mathparsmall}
    \ebrule{
    \hypo{i:\SortLvl}
    \hypo{M,N:\SortTm\prn{\TpUniv_i}}
    \hypo{
      \TpEl_i\prn{M} = \TpEl_i\prn{N} : \SortTp
    }
    \infer[admissible]3{
      M = N : \SortTm\prn{\TpUniv_i}
    }
  }
  \end{mathparsmall}
\end{corollary}

\subsection{Corollary: Admissible Equivalence of \texorpdfstring{\ThyFussFree}{FussFree} and \texorpdfstring{\ThyMLTTNat}{NatHier}}

As a consequence of \zcref{thm:inj-decoding}, the initial model of \ThyMLTTNatInj{}, if regarded as a model of \ThyMLTTNat{}, is in fact the \emph{initial} model of the latter. This can be thought of as an invariant version of \emph{admissibility} that can be stated in categorical semantics without reference to any particular presentation by inference rules; from a concrete perspective, it means that syntax written in either language can be translated back and forth coherently:
\[\begin{tikzcd}
	\ThyMLTTNat && \ThyMLTTNatInj && \ThyFussFree
	\arrow["{\text{Admissible}}", Leftrightarrow, from=1-3, to=1-1]
	\arrow["{\text{Derivable}}"', Leftrightarrow, from=1-3, to=1-5]
\end{tikzcd}\]

The equivalence above pertains only to syntax: the two theories do \emph{not} have the same models. Syntactic equivalences like this are very useful: you can write your code in the fuss-free calculus, translate it into the natural hierarchy calculus, and then interpret it in any model of the latter.

\section{Bidirectional Elaboration of the Fuss-Free Hierarchy}\label{sec:elaboration}

In a practical system, the user should not have to manipulate the decoding, lifts or codes themselves;
our goal is to make working with universe levels as simple as possible, and to minimise the number of
places where they need to appear. This means that the raw syntax (\emph{i.e.}\ what the user writes) should be
as \textit{implicit} as possible, and it is natural in this case to conflate types and terms as a matter of notation so that a term $A : \TpUniv_i$ may also be used as a proper type, and a type that happens to be small can be used as an element of a suitable universe.

The usual way to achieve this is to \textit{elaborate} the raw syntax to the core theory \ThyFussFree, by inserting
the administrative ``coding'' ($\Shrink{}$) and ``decoding'' ($\Enlarge{}$) operations automatically everywhere they are needed. It would of course be possible to take \ThyMLTTNatInj{} as the target of elaboration, but \ThyFussFree{} is considerably simpler to implement and is justified by our metatheoretic results (\zcref{sec:normalisation}). In particular, an implementation of \ThyFussFree{} involves far fewer term and type constructors, leading to a streamlined treatment of conversion relative to the more bureaucratic \ThyMLTTNatInj.

\paragraph{Raw Syntax}
To describe the bidirectional elaboration algorithm, we divide the raw syntax into two different sorts: the
\textit{synthesizable} and \textit{checkable} terms. The former will be the subject of the type synthesis problem, which
\textit{infers} the type of such a given term. The latter will be the subject of the type checking problem, which checks if a
term has a given type in a given context.

\begin{tabular-bnf}
    Levels & \raw{i},\raw{j} ::= & \raw{\alpha}\mid \raw{0}\mid \raw{i^+} \\
    Synthesis & \raw{E},\raw{F} ::= & \raw{x} \mid \raw{E}\,\raw{M} \mid \raw{\ConcLvlApp{E}{i}} \mid \raw{\TmFst} \raw{E} \mid \raw{\TmSnd} \raw{E}\\
    Checking & \raw{M},\raw{N},\raw{A},\raw{B} ::= & \raw{E} \mid \raw{\lambda x\mathpunct{.} M} \mid \raw{\ConcLvlLam{\alpha}{M}} \mid \raw{(M,N)} \mid \raw{(x:A)\to B} \mid \raw{\ConcLvlPi{\alpha}{A}} \mid \raw{(x:A)\times B} \mid \raw{\TpUniv_i}
\end{tabular-bnf}

Here, raw universe levels $\raw{i}, \raw{j}$ are built from level variables $\raw{\alpha}$, the base level $\raw{0}$, and the successor operation $\raw{i^+}$.
Synthesizable terms $\raw{E}, \raw{F}$ comprise term variables $\raw{x}$, function application $\raw{E}\,\raw{M}$, level application $\raw{\ConcLvlApp{E}{i}}$, and the first and second projections $\raw{\TmFst E}, \raw{\TmSnd E}$.
Checkable expressions include embedded synthesizable terms $\raw{E}$, function abstraction $\raw{\lambda x\mathpunct{.} M}$, level abstraction $\raw{\ConcLvlLam{\alpha}{M}}$, dependent pairs $\raw{(M,N)}$, dependent function types $\raw{(x:A)\to B}$, level-polymorphic product types $\raw{\ConcLvlPi{\alpha}{A}}$, dependent pair types $\raw{(x:A)\times B}$, and universes $\raw{\TpUniv_i}$.

\paragraph{Resolver Environments}
Along with this raw syntax, we will rely on \textit{resolver environments} $\core{\varrho}$ to handle the surface level variable names written by
the user. A resolver environment $\core{\varrho} : \VarIdent \rightharpoonup \CoreTypedTerm + \SortLvl$ is a partial map from the raw syntactic identifiers
$\raw{x} \in \VarIdent$ to core typed terms $\prn{\core{M}\in\core{A}}$, and from raw level identifiers $\raw{\alpha}$ to core levels $\prn{\core{i}\in \SortLvl}$. Resolver environments are always considered with respect to a well-formed context $\core{\Gamma}$, and must
satisfy the following invariant: $\core{\varrho}$ is said to be \textit{valid} with respect to $\core{\Gamma}$ if for every identifier $\raw{x}$ in the domain of $\core{\varrho}$ we have:
$\core{\varrho}\prn{\raw{x}} \equiv \prn{\core{M} \in \core{A}} \implies \core{\Gamma} \vdash \core{M : A}$ and 
$\core{\varrho}\prn{\raw{\alpha}} \equiv \prn{\core{i}\in\SortLvl} \implies \core{\Gamma} \vdash \core{i}: \SortLvl$.

\paragraph{Elaboration Judgements}
Now we can describe the bidirectional algorithm itself, by defining four mutually recursive judgements.
\begin{enumerate}
    \item \textbf{Level synthesis:} The judgement $\core{\Gamma};\core{\varrho}\vdash\raw{i}\syn\core{i'}\in\SortLvl$ takes a context $\core{\Gamma}$, a $\core{\Gamma}$-valid environment $\core{\varrho}$, and a raw level $\raw{i}$ as inputs, and outputs a core level $\core{\Gamma}\vdash\core{i'}:\SortLvl$.
    \item \textbf{Term synthesis:} The judgement $\core{\Gamma};\core{\varrho} \vdash \raw{E} \syn \core{M} \in \core{A}$ takes $\core{\Gamma}$, $\core{\varrho}$, and a synthesis-mode raw term $\raw{E}$ as inputs, and outputs both an elaborated core term $\core{M}$ and its inferred core type $\core{A}$ such that $\core{\Gamma}\vdash \core{M}:\core{A}$.
    \item \textbf{Term checking:} The judgement $\core{\Gamma};\core{\varrho} \vdash \core{A} \ni \raw{M} \chk \core{M'}$ takes $\core{\Gamma}$, $\core{\varrho}$, an expected core type $\core{A}$ (with $\core{\Gamma} \vdash \core{A}\ \mathit{type}$), and a checking-mode raw term $\raw{M}$ as inputs, and outputs a core term $\core{\Gamma}\vdash \core{M'}:\core{A}$. Writing $\core{A} \ni \raw{M}$ places the known type $\core{A}$ before $\raw{M}$ to emphasise that $\core{A}$ is an input, in contrast to $\core{M} \in \core{A}$ in synthesis where $\core{A}$ is an output.
    \item \textbf{Type checking:} The judgement $\core{\Gamma};\core{\varrho} \vdash \raw{A} \chktp[i^?] \core{A'}$ takes $\core{\Gamma}$, $\core{\varrho}$, a checking-mode raw type $\raw{A}$, and an optional core level bound $\core{i^?}$ as inputs, and outputs a well-formed core type $\core{\Gamma}\vdash \core{A'}\ \mathit{type}$. When the optional parameter $\core{i^?}$ is provided as a level $\core{\Gamma}\vdash \core{i}:\SortLvl$, the elaborated type is guaranteed to be $i$-small, i.e.\ $\core{\Gamma}\vdash \SortSmall_i\prn{\core{A'}}$ holds.
\end{enumerate}

Given a core level $\core{j}$ and an optional core level $\core{i^?}$, we write $\core{j}\leq\core{i^?}$ to mean $\core{j}\leq\core{i}$ when $\core{i^?} = \core{i}$ is defined; otherwise (when $\core{i^?}$ is absent), the inequality holds by fiat.

Our main addition to standard bidirectional type checking algorithms~\citep{coquand-type-checking} is the optional size parameter $\core{i^?}$ of the judgement $\mathsf{chktp}$.
Concretely, this judgement corresponds to a bidirectional version of the formation rules, and the size parameter embeds the smallness predicate.
This contrasts with the bidirectional type checking of cumulative Russell-style universes in Lennon-Bertrand's PhD thesis~\citep{lennon-bertrand-2021,lennon-bertrand-phd}. Rocq lacks a proper sort of \emph{all} types, so type well-formedness is checked via a \emph{constrained inference} mode ($\triangleright_{\mathsf{Sort}}$) that synthesises a universe level bottom-up.
In contrast, our algorithm elaborates an implicit source-level syntax into an explicit algebraic core by propagating size bounds downward.

\begin{mathpartiny}
  \ebrule[syn-hyp]{
    \hypo{
      \core{\varrho}\prn{\raw{x}} \equiv \prn{\core{M}\in\core{A}}
    }
    \infer1{
      \core{\Gamma};\core{\varrho} \vdash \raw{x} \syn \core{M} \in \core{A}
    }
  }
  \and
  \ebrule[syn-lvl-hyp]{
    \hypo{
      \core{\varrho}\prn{\raw{\alpha}} \equiv \prn{\core{i}\in\SortLvl}
    }
    \infer1{
      \core{\Gamma};\core{\varrho}\vdash \raw{\alpha}\syn \core{i}\in\SortLvl
    }
  }
  \and
  \ebrule[syn-lvl-zero]{
    \infer0{
      \core{\Gamma};\core{\varrho}\vdash\raw{0}\syn \core{0}\in\SortLvl
    }
  }
  \and
  \ebrule[syn-lvl-succ]{
    \hypo{
      \core{\Gamma};\core{\varrho}\vdash\raw{i}\syn \core{i'}\in\SortLvl
    }
    \infer1{
      \core{\Gamma};\core{\varrho}\vdash\raw{i^+}\syn \core{i'^+}\in\SortLvl
    }
  }
  \and
  \ebrule[syn-app]{
    \hypo{
      \core{\Gamma};\core{\varrho} \vdash \raw{E} \syn \core{E'} \in \core{\TpPi\prn{A,B}}
    }
    \hypo{
      \core{\Gamma};\core{\varrho} \vdash \core{A} \ni \raw{N} \chk \core{N'}
    }
    \infer2{
      \core{\Gamma};\core{\varrho} \vdash \raw{E\,N} \syn \core{\TmApp\brc{A,B}\prn{E',N'}} \in \core{B\brk{N'}}
    }
  }
  \and
  \ebrule[syn-lvl-app]{
    \hypo{
      \core{\Gamma};\core{\varrho} \vdash \raw{E} \syn \core{E'} \in \core{\TpLvlPi\prn{A}}
    }
    \hypo{
      \core{\Gamma};\core{\varrho}\vdash \raw{i} \syn \core{i'}\in\SortLvl
    }
    \infer2{
      \core{\Gamma};\core{\varrho} \vdash \raw{\ConcLvlApp{E}{i}} \syn \core{\TmLvlApp\brc{A}\prn{E',i'}} \in \core{A\brk{i'}}
    }
  }
  \and
  \ebrule[syn-proj1]{
    \hypo{
      \core{\Gamma};\core{\varrho} \vdash \raw{E} \syn \core{E'} \in \core{\TpSigma\prn{A,B}}
    }
    \infer1{
      \core{\Gamma};\core{\varrho} \vdash \raw{\TmFst E} \syn \core{\TmFst\brc{A,B}\prn{E'}} \in \core{A}
    }
  }
  \and
  \ebrule[syn-proj2]{
    \hypo{
      \core{\Gamma};\core{\varrho} \vdash \raw{E} \syn \core{E'} \in \core{\TpSigma\prn{A,B}}
    }
    \infer1{
      \core{\Gamma};\core{\varrho} \vdash \raw{\TmSnd E} \syn \core{\TmSnd\brc{A,B}\prn{E'}} \in \core{B\brk{\TmFst\brc{A,B}\prn{E'}}}
    }
  }
  \and
  \ebrule[chk-code]{
    \hypo{\core{\Gamma};\core{\varrho} \vdash \raw{A} \chktp[i] \core{A'}}
    \infer1{\core{\Gamma};\core{\varrho} \vdash \core{\TpUniv_i} \ni \raw{A} \chk \core{\Shrink{i}{A'}} }
  }
  \and
  \ebrule[chk-pair]{
    \hypo{
      \core{\Gamma}; \core{\varrho} \vdash \core{A} \ni \raw{M} \chk \core{M'}
    }
    \hypo{
      \core{\Gamma,x:A}; \core{\varrho}\brk{\raw{x}\mapsto \core{x}\in\core{A}} \vdash \core{B\brk{x}} \ni \raw{N} \chk \core{N'\brk{x}}
    }
    \infer2{
      \core{\Gamma};\core{\varrho} \vdash \core{\TpSigma\prn{A,B}}
      \ni
      \raw{\prn{M,N}} \chk \core{\TmPair\brc{A,B}\prn{M',N'}}
    }
  }
  \and
  \ebrule[chk-lambda]{
    \hypo{
      \core{\Gamma,x:A}; \core{\varrho}\brk{\raw{x}\mapsto \core{x}\in\core{A}} \vdash \core{B\brk{x}} \ni \raw{M} \chk \core{M'\brk{x}}
    }
    \infer1{
      \core{\Gamma};\core{\varrho} \vdash \core{\TpPi\prn{A,B}}
      \ni
      \raw{\lambda x\mathpunct{.}M} \chk \core{\TmLam\brc{A,B}\prn{M'}}
    }
  }
  \and
  \ebrule[chk-lambda-lvl]{
    \hypo{
      \core{\Gamma,\alpha:\SortLvl}; \core{\varrho}\brk{\raw{\alpha}\mapsto \core{\alpha}\in\core{\SortLvl}} \vdash \core{A\brk{\alpha}} \ni \raw{M} \chk \core{M'\brk{\alpha}}
    }
    \infer1{
      \core{\Gamma};\core{\varrho} \vdash \core{\TpLvlPi\prn{A}}
      \ni
      \raw{\ConcLvlLam{\alpha}{M}} \chk \core{\TmLvlLam\brc{A}\prn{M'}}
    }
  }
  \and
  \ebrule[switch]{
    \hypo{\core{\Gamma};\core{\varrho} \vdash \raw{E} \syn \core{E'} \in \core{B}}
    \hypo{\core{\Gamma}\vdash \coe \core{E'} : \core{B} \Rightarrow \core{A} \rightsquigarrow \core{E''}}
    \infer2{\core{\Gamma};\core{\varrho} \vdash \core{A} \ni \raw{E} \chk \core{E''}}
  }
  \and
  \ebrule[chktp-univ]{
    \hypo{\core{\Gamma};\core{\varrho} \vdash \raw{i} \syn \core{i'} \in \SortLvl}
    \hypo{\core{\Gamma} \vdash \core{i'^+} \leq \core{j^?}}
    \infer2{\core{\Gamma};\core{\varrho} \vdash \raw{\TpUniv_i} \chktp[j^?] \core{\TpUniv_{i'}} }
  }
  \and
  \ebrule[chktp-pi]{
    \hypo{\core{\Gamma};\core{\varrho} \vdash \raw{A} \chktp[i^?] \core{A'} }
    \hypo{\core{\Gamma,x:A'};\core{\varrho}\brk{\raw{x}\mapsto \prn{\core{x}\in \core{A'}}} \vdash \raw{B} \chktp[i^?] \core{B'\brk{x}} }
    \infer2{\core{\Gamma};\core{\varrho} \vdash \raw{(x:A)\to B} \chktp[i^?] \core{\TpPi(A',B')} }
  }
  \and
  \ebrule[chktp-pi-lvl]{
    \hypo{\core{\Gamma,\alpha:\SortLvl};\core{\varrho}\brk{\raw{\alpha}\mapsto \prn{\core{\alpha}\in \core{\SortLvl}}} \vdash \raw{A} \chktp[i^?] \core{A'\brk{\alpha}} }
    \infer1{\core{\Gamma};\core{\varrho} \vdash \raw{\ConcLvlPi{\alpha}{A}} \chktp[i^?] \core{\TpLvlPi\prn{A'}} }
  }
  \and
  \ebrule[chktp-sg]{
    \hypo{\core{\Gamma};\core{\varrho} \vdash \raw{A} \chktp[i^?] \core{A'} }
    \hypo{\core{\Gamma,x:A'};\core{\varrho}\brk{\raw{x}\mapsto \prn{\core{x}\in \core{A'}}} \vdash \raw{B} \chktp[i^?] \core{B'\brk{x}} }
    \infer2{\core{\Gamma};\core{\varrho} \vdash \raw{(x:A)\times B} \chktp[i^?] \core{\TpSigma\prn{A',B'}} }
  }
  \and
  \ebrule[switch-ty]{
    \hypo{\core{\Gamma};\core{\varrho} \vdash \raw{E} \syn \core{E'} \in \core{\TpUniv_j} }
    \hypo{
      \core{\Gamma}\vdash \core{j}\leq \core{i^?}
    }
    \infer2{
      \core{\Gamma};\core{\varrho} \vdash \raw{E} \chktp[i^?] \core{\Enlarge{j}{E'}}
    }
  }
\end{mathpartiny}

To ensure that the algorithm is syntax-directed, we restrict \textsc{chk-code} to apply only to subjects $\raw{A}$ with a canonical type head-constructor (i.e.\ $\raw{(x:A)\to B}$, $\raw{\ConcLvlPi{\alpha}{A}}$, $\raw{(x:A)\times B}$, or $\raw{\TpUniv_j}$). In fact, this amounts to prioritising \textsc{switch} over \textsc{chk-code} when checking a synthesizable term $\raw{E}$ against $\core{\TpUniv_i}$.

\paragraph{Subtyping Coercions}
We now describe the coercion judgement $\Gamma \vdash \coe \core{M} : \core{B} \Rightarrow \core{A} \rightsquigarrow \core{M'}$, which coerces
the term $\core{M}$ of type $\core{B}$ to a term $\core{M'}$ of type $\core{A}$, and falls back to the conversion checker in the appropriate cases.
\begin{mathpartiny}
  \ebrule[coe-pi]{
    \hypo{
      \begin{array}{l}
        \core{\Gamma,x:A'} \vdash \coe \core{x}: \core{A'} \Rightarrow \core{A} \rightsquigarrow \core{u\brk{x}}
        \\
        \core{\Gamma,x:A'} \vdash \coe \core{\TmApp\brc{A,B}\prn{M,u\brk{x}}}: \core{B\brk{u\brk{x}}} \Rightarrow \core{B'\brk{x}} \rightsquigarrow \core{N\brk{x}}
      \end{array}
    }
    \infer1{\core{\Gamma} \vdash \coe \core{M} : \core{\TpPi(A,B)} \Rightarrow \core{\TpPi(A',B')} \rightsquigarrow \core{\TmLam\brc{A',B'}\prn{N}}}
  }
  \and
  \ebrule[coe-sg]{
    \hypo{
      \begin{array}{l}
        \core{\Gamma} \vdash \coe \core{\TmFst\brc{A,B}\prn{M}}: \core{A} \Rightarrow \core{A'} \rightsquigarrow \core{N_0}
        \\
        \core{\Gamma} \vdash \coe \core{\TmSnd\brc{A,B}\prn{M}}: \core{B\brk{\TmFst\brc{A,B}\prn{M}}} \Rightarrow \core{B'\brk{N_0}} \rightsquigarrow \core{N_1}
      \end{array}
    }
    \infer1{\core{\Gamma} \vdash \coe \core{M} : \core{\TpSigma\prn{A,B}} \Rightarrow \core{\TpSigma\prn{A',B'}} \rightsquigarrow \core{\TmPair\brc{A',B'}\prn{N_0, N_1}}}
  }
  \and
  \ebrule[coe-lvl-pi]{
    \hypo{
      \core{\Gamma,\alpha:\SortLvl} \vdash \coe \core{\TmLvlApp\brc{A}\prn{M,\alpha}}: \core{A\brk{\alpha}} \Rightarrow \core{A'\brk{\alpha}} \rightsquigarrow \core{N\brk{\alpha}}
    }
    \infer1{\core{\Gamma} \vdash \coe \core{M} : \core{\TpLvlPi\prn{A}} \Rightarrow \core{\TpLvlPi\prn{A'}} \rightsquigarrow \core{\TmLvlLam\brc{A'}\prn{N}}}
  }
  \\
  \ebrule[coe-lift]{
    \hypo{\strut\core{\Gamma} \vdash \core{i} \leq \core{j}}
    \infer1{\core{\Gamma} \vdash \coe \core{M} : \core{\TpUniv_i} \Rightarrow \core{\TpUniv_j} \rightsquigarrow\core{\Shrink{j}\Enlarge{i} M}}
  }
  \and
  \ebrule[coe-conv]{
    \hypo{\strut\core{\Gamma} \vdash \core{A} \equiv \core{B}}
    \infer1{\core{\Gamma} \vdash \coe \core{M} : \core{A} \Rightarrow \core{B} \rightsquigarrow\core{M}}
  }
\end{mathpartiny}

\section{Tutorial Implementation of Cumulative Universes}\label{sec:tutorial-implementation}

We now describe how to implement a theory with a cumulative universe hierarchy, by taking advantage of the
fuss-free presentation and the algorithm described in \zcref{sec:elaboration}. The running example is a minimal implementation in Haskell that takes a prior implementation by \citet{elabzoo,smalltt} as its starting point.
Our implementation includes of the fuss-free datatypes from \zcref{sec:fuss-free-datatypes} with a user-friendly surface syntax \`a la \citet{dagand:2013}, along with record types. Our implementation of datatypes differs from the discussion in \zcref{sec:cumulative-inductive-types} in one way: in order to allow the definition of datatypes like $\mathsf{List}$ that are parameterised in other types, we extend the universe hierarchy by a transfinite universe level rather than introducing schematic constructors, which would complicate the implementation.

\subsection{Evaluation and Conversion Checking}

As described in the introduction, our fuss-free approach offers some simplifications to the implementation of a type checker's core, \emph{e.g.}\ in the definition of the semantic domain and the evaluation algorithm.
In the following, we describe in more detail how to implement a simple evaluation algorithm for the theory \ThyFussFree.\footnote{For the sake of pedagogy, we do not treat the most optimised possible implementation, which would employ Kov\'acs' \emph{glued evaluation}~\citep{smalltt}; adapting glued evaluation to the fuss-free universes is no more difficult than for pure MLTT.}

\begin{lstlisting}
newtype Ix = Ix Int deriving (Eq, Show, Num) via Int
type Name = String
\end{lstlisting}

The user-facing raw syntax described in \zcref{sec:elaboration} corresponds to the following datatype \texttt{Raw}.
We will elaborate this to the core theory \ThyFussFree, which is implemented as follows (\texttt{Tp, Tm}). We deviate from the fully algebraic syntax in a few ways: and we remember the concrete names of variables (for printing error messages). The missing annotations in the core representations are present in the trace of the bidirectional elaboration algorithm, but it is not efficient to explicitly store them.

\noindent
\begin{minipage}[t]{.5\textwidth}
\begin{lstlisting}
data RawSize 
  = RSzVar Name | RSz Int
  | RSucc RawSize deriving Show
data Raw
  = RVar Name              $\qquad \raw{x}$
  | RPi Name Raw Raw       $\qquad \raw{(x : A) \to B}$
  | RLam Name Raw          $\qquad \raw{\lambda x\mathpunct{.}M}$
  | RApp Raw Raw           $\qquad \raw{E}\;\raw{M}$
  | RU RawSize             $\qquad \raw{\TpUniv_i}$
  | RLPi Name Raw          $\qquad \raw{\ConcLvlPi{\alpha}{A}}$
  | RLAbs Name Raw         $\qquad \raw{\ConcLvlLam{\alpha}{M}}$
  | RLApp Raw RawSize      $\qquad \raw{\ConcLvlApp{M}{i}}$
\end{lstlisting}
\end{minipage}%
\begin{minipage}[t]{.5\textwidth}
\begin{lstlisting}
data Size = LVar Ix | Zero | Succ Size
data Tp
  = Pi Name ~Tp Tp         $\qquad \core{\TpPi\prn{A,B}}$
  | U Size                 $\qquad \core{\TpUniv_i}$
  | Decode Size Tm         $\qquad \core{\Enlarge{i}{M}}$
  | LPi Name Tp            $\qquad \core{\ConcLvlPi{\alpha}{A}}$
data Tm
  = Var Ix                 $\qquad \core{x}$
  | App Tm ~Tm             $\qquad \core{\TmApp\brc{A,B}\prn{M,N}}$
  | Lam Name Tm            $\qquad \core{\TmLam\brc{A,B}\prn{M}}$
  | Code Size Tp           $\qquad \core{\Shrink{i}{A}}$
  | LAbs Name Tm           $\qquad \core{\TmLvlLam\brc{A}\prn{M}}$
  | LApp Tm Size           $\qquad \core{\TmLvlApp\brc{A}\prn{M,i}}$
\end{lstlisting}
\end{minipage}

Following \citet{abel-2013}, we use de Bruijn \emph{levels} for variables in the semantic domains, which counts in the opposite way to indices.
This technicality allows us to define an optimised algorithmsithm that does not rely on any lifting of de Bruijn indices, with an helper function \texttt{lvl2Ix}.

\begin{lstlisting}
newtype DeBruijnLvl = DeBruijnLvl Int deriving (Eq, Show, Num) via Int
lvl2Ix (DeBruijnLvl l) (DeBruijnLvl x) = Ix (l - x - 1)
\end{lstlisting}

The core terms (\texttt{Tm}) and types (\texttt{Tp}) will evaluate to the semantic domains \texttt{VTp} and \texttt{VTm}, defined by the following datatypes.

\begin{lstlisting}
data VSize = VLVar DeBruijnLvl | VZero | VSucc VSize
data VTp = VPi Name ~VTp (VTm -> VTp) | VU VSize | VDecode VSize VTm | VLPi Name (VSize -> VTp)
data VTm 
  = VVar DeBruijnLvl | VApp VTm ~VTm | VLam Name (VTm -> VTm)
  | VCode VSize VTp | VLAbs Name (VSize -> VTm) | VLApp VTm VSize
\end{lstlisting}

Normalisation by Evaluation (NbE) amounts to writing an evaluation function from the syntax to this semantic domain, along
with a quotation function which maps back values to the corresponding normal forms in the syntax. Any term $\core{t}$ will be evaluated
with respect to an environment (\texttt{Env}), which is a list of semantic terms providing values for the free variables.

\noindent
\begin{minipage}{.5\textwidth}
\begin{lstlisting}
evalSize :: Env -> Size -> VSize
evalTp :: Env -> Tp -> VTp
quoteTp :: DeBruijnLvl -> VTp -> Tp
convTp :: DeBruijnLvl -> VTp -> VTp -> Bool
\end{lstlisting}
\end{minipage}%
\begin{minipage}{.5\textwidth}
\begin{lstlisting}
quoteSize :: DeBruijnLvl -> VSize -> Size
evalTm :: Env -> Tm -> VTm
quoteTm :: DeBruijnLvl -> VTm -> Tm
convTm :: DeBruijnLvl-> VTm -> VTm -> Bool
\end{lstlisting}
\end{minipage}

We now describe the evaluation functions, emphasizing rules associated to universes.
The functions \texttt{quoteTm} and \texttt{quoteTp} and the conversion algorithm are rather standard, and no real complication is added by the Fuss-Free presentation.

\begin{snippet}[\texttt{evalTm} and \texttt{evalTp}]
We present the rules associated to level polymorphism, which are almost identical to the standard rules for dependent products and lambda abstractions. To evaluate a decoding, we apply the $\beta$-rule for the universe if the size matches; this is a key point at which our algorithm deviates from NbE for non-fuss-free universes.

\noindent
\begin{minipage}{.5\textwidth}
\begin{lstlisting}
evalTm env = \case
  {- ... -}
  LAbs x t -> VLAbs x \i ->
    evalTm (VESize i : env) t
  LApp t s -> case evalTm env t of
    VLAbs _ f -> f (evalSize env s)
    f -> VLApp f (evalSize env s)
\end{lstlisting}
\end{minipage}%
\begin{minipage}{.5\textwidth}
\begin{lstlisting}
evalTp env = \case
  {- ... -}
  LPi x t -> VLPi x \i ->
    evalTp (VESize i : env) t
  Decode i t -> case evalTm env t of
    VCode j a -> a  
    v -> VDecode (evalSize env i) v
\end{lstlisting}
\end{minipage}
\end{snippet}

\subsection{Elaboration in Practice}

We now describe the type checking algorithm in a detailed way.
The judgements of \zcref{sec:elaboration} correspond to mutually defined recursive functions, with signatures as follows.
\begin{lstlisting} 
infer :: Cxt -> Raw -> M (Tm, VTp)
check :: Cxt -> Raw -> VTp -> M Tm
checkTp :: Cxt -> Raw -> Maybe Size -> M Tp
coe :: Cxt -> DeBruijnLvl -> VTp -> VTp -> Tm -> M Tm
\end{lstlisting}

Here, \texttt{M} is the type checking monad, which handles the two possible outcomes of the elaboration algorithm: either we return a pure term/type, or an error annotated with a source position.

Along with this, we need a function \texttt{elabSize} which elaborates raw levels $\raw{i}$.
\begin{lstlisting}
elabSize :: Cxt -> RawSize -> M Size
\end{lstlisting}

More concretely, the elaboration rules of \zcref{sec:elaboration} translate to the different match cases defining the corresponding recursive functions.
If we compare to standard type checking algorithms~\citep{coquand-type-checking}, our main structural difference is the optional size parameter of the function \texttt{checkTp}.
Level polymorphism is handled rather easily, with very similar rules to the ones for regular products.

\begin{snippet}[\texttt{checkTp}]
The rules defining the judgement $\core{\Gamma};\core{\varrho} \vdash \raw{A} \chktp[i^?] \core{A'}$ give rise to the following recursive function \texttt{checkTp}.
Note that \texttt{size} is either \texttt{Nothing} or \texttt{Just sz}, depending on whether we need to check the type to be small or not, as explained in \zcref{sec:elaboration}.
\begin{lstlisting}[mathescape=false]
checkTp cxt t size = case t of
  {- ... -}
  RU s -> do
    s <- elabSize cxt s
    case size of
      Nothing -> pure $ U s
      Just sz -> do
        let vs = evalSize (env cxt) s 
        if vs < sz then pure $ U s
        else report cxt "Size issue"
  _ -> do
    (t, a) <- infer cxt t
    case a of
      VU i -> case size of
        Nothing -> pure (Decode (quoteSize (lvl cxt) i) t)
        Just sz -> do
          if i <= sz then pure (Decode (quoteSize (lvl cxt) i) t)
          else report cxt "Size issue"
      _ -> report cxt "Elaboration error: Expected a code"
\end{lstlisting}
\end{snippet}

The functions \texttt{check} and \texttt{infer} are very similar to standard bidirectional type checking algorithms, except as pertains to universes and levels.

\begin{snippet}[\texttt{check} and \texttt{infer}]
The rules defining the checking and synthesis judgements  $\core{\Gamma};\core{\varrho} \vdash \core{A} \ni \raw{M} \chk \core{M'}$ and $\core{\Gamma};\core{\varrho} \vdash \raw{E} \syn \core{M} \in \core{A}$ give rise to the following recursive functions.

\noindent
\begin{minipage}[t]{.5\textwidth}
\begin{lstlisting}[mathescape=false]
check cxt t a = case (t, a) of
  {- ... -}
  (_, VU i) -> do
    u <- checkTp cxt t (Just i)
    pure $ Code (quoteSize (lvl cxt) i) u
  _ -> do
    (m, bTp) <- infer cxt t
    coe cxt (lvl cxt) bTp a m
\end{lstlisting}
\end{minipage}%
\begin{minipage}[t]{.5\textwidth}
\begin{lstlisting}[mathescape=false]
infer cxt = \case
  {- ... -}
  RLApp t s -> do
    (tTm, tTp) <- infer cxt t
    case tTp of
      VLPi _ b -> do
        sz <- elabSize cxt s
        let vsz = evalSize (env cxt) sz
        pure (LApp tTm sz, b vsz)
      _ -> report cxt "Expected level quantifier"
\end{lstlisting}
\end{minipage}
\end{snippet}

The mode switch rule of the function \texttt{check} calls the function \texttt{coe}, which we now describe.

\begin{snippet}[\texttt{coe}]
The function \texttt{coe} corresponds to the operation $\coe \core{E} : \core{B} \Rightarrow \core{A} \rightsquigarrow \core{E'}$.

\begin{lstlisting}[mathescape=false]
coe cxt l sourceTp targetTp m = case (sourceTp, targetTp) of
  {- ... -}
  (VU i, VU j) | i <= j -> pure $ Code j (Decode i m)
  _ -> if convTp l sourceTp targetTp then pure m else report cxt "Error: Invalid coercion"
\end{lstlisting}
\end{snippet}

\subsection{Implementation of Fuss-Free Datatypes}\label{sec:impl-datatypes}

We now extend the implementation to include the fuss-free inductive datatypes of \zcref{sec:fuss-free-datatypes}.

\paragraph{Descriptions, Elaboration}
Following \citet{dagand:2013} and the theory presented in \zcref{sec:fuss-free-datatypes}, the core syntax and semantic domain are extended with datatype descriptions (\texttt{Desc} and \texttt{VDesc}), the initial algebra type \texttt{Mu}, its constructor \texttt{In}, and the eliminator \texttt{Elim}.
Our elaborator translates user-written declarations like the following into core datatype descriptions as described by \citet{dagand:2013}:
\begin{lstlisting}
data T params : U_i where { c_1 : T_1; ...; c_k : T_k };
\end{lstlisting}

\paragraph{Stratifying Large Types via \texttt{Big}.}
We extend the level type \texttt{Size} with a top level $\mathtt{Big}$:
\begin{lstlisting}
data Size  = LVar Ix | Zero | Succ Size | Big
data VSize = VLVar DeBruijnLvl | VZero | VSucc VSize | VBig 
\end{lstlisting}
Both the level successor operation and level-polymorphic instantiation need to be restricted to small levels by using  a predicate \texttt{isSmallSize}.
The top level \texttt{Big} represents a large universe $\TpUniv_{\mathtt{Big}}$ containing all small universes $\TpUniv_i$.
Users can thus declare parameterised datatypes at level \texttt{Big}, and later shrink their instances down to any small universe $\TpUniv_i$ at the use site.

\paragraph{Shrinking Large Datatypes on Demand}
To extend the fuss-free approach to datatypes, one has to implement their smallness predicate $\SortSmallD$ as described in \zcref{sec:fuss-free-datatypes}.
Although it is semantically convenient to define smallness of descriptions in terms of their extension, we opted for a more direct implementation by recursion on the structure of the description.

\begin{lstlisting}
isSmall :: Cxt -> VSize -> VTy -> M ()
isSmall cxt i = \case
  VMu _ d  -> isDescSmall cxt i d
  {- ... -}

isDescSmall :: Cxt -> VSize -> VDesc -> M ()
isDescSmall cxt i = \case
  VDescTensor d1 d2 -> isDescSmall cxt i d1 >> isDescSmall cxt i d2
  VDescSum n a d -> isSmall cxt i a >> isDescSmall (bind n a cxt) i (d (VVar (lvl cxt)))
  VDescProd n a d -> isSmall cxt i a >> isDescSmall (bind n a cxt) i (d (VVar (lvl cxt)))
  {- ... -}
\end{lstlisting}
\section{Conclusions and Future Work}

We have contributed the \emph{fuss-free} presentation of cumulative universes, which we believe is the closest thing possible to an algebraic version of Martin-L\"of's universes \`a la Russell extended by universe polymorphism. We proved that system equivalent to a more intricate theory with coherent level coercions in the style of \citet{bcde-2022} by proving the normalisation theorem for the latter. We have seen that the fuss-free presentation scales effortlessly to an account of cumulative inductive datatypes via a judgemental version of datatype descriptions. We believe that our design is appropriate for implementation in future dependently typed programming languages and proof assistants.
There are a few areas where further work could be beneficial.

\emph{Assessing Usability.}
Our fuss-free presentation of universes is motivated by our extensive personal experience in the design, implementation, metatheory, and use of interactive proof assistants; only time will tell whether our approach leads in fact to improved usability and reliability. Our claim is more narrow: (1) the fuss-free presentation recovers the user-visible behaviour of implicit cumulative Russellian universes in a nonetheless generalised algebraic type system, and (2) it does so without introducing a proliferation of universe levels (\emph{cf.}\ \citet{coquand-2012-presheaf-model}), and (3) it may be implemented in a purely syntax-directed way that bypasses the challenging inference tasks associated with more ambitious proposals (\emph{cf.}\ \citet{bezem-sozeau-2025}).
We note that two experimental implementations of dependent type theory have already adopted our approach: Hippogriff~\citep{lynch-2026} and Project Pterodactyl~\citep{sterling-2025}.
Scaling these prototypes to larger production-quality implementations will be necessary to assess the practical impact in terms that go beyond the qualitative.

\emph{Unification and Coercions.}
It is well known that higher-order unification and subtyping coercions do not play well together in existing systems; the issue has to do with \emph{scheduling}. For example, solving a metavariable could be the thing that makes a type small, and this could make elaboration sensitive to the order in which cumulativity coercions and unification metavariables are solved.

As ever, the solution must be a principled account of the \emph{postponement} of both coercion and unification problems. We suggest dealing with those two problems in a \textit{uniform} way, by formally treating unification as a special case of coercion.
In recent work, \citet{kovacs-observing-equality} has introduced a new algorithm to deal with postponed problems using \emph{observational} coercions.
We believe that extending the observational coercions of Kov\'acs' algorithm to handle subtyping coercions would improve the behaviour of such a system. 
Relating subtyping and observational type theory closely relates to the work of \citet{adaptt}, which gives a unifying framework for functorial casts. Whether the full power of \emph{op.\ cit.}\ is required remains to be seen.

\emph{Models of Cumulativity.}
We have focused entirely on syntax, but even if it is internally consistent, there is little use in a type theory that cannot be interpreted in a \emph{variety} of different categories. For non-univalent type theory, the situation is very good: using the results of \citet{gratzer-shulman-sterling-2026}, we have direct models of \ThyFussFree{} in any Grothendieck topos (including sheaf toposes). For higher toposes, the matter is slightly murkier: \citet{shulman-2019} describes a model construction in which the universe lifting coercions are not injective, but it seems very likely that by unravelling his 2-dimensional small object argument more explicitly we would be able to choose injective lifting coercions.\footnote{One might have hoped instead to factor the interpretation through the equivalence between \ThyMLTTNat{} and \ThyFussFree{}, but the usual way to obtain codes for connectives that are natural with respect to universe levels is to use \emph{realignment}---which depends on the coercions already being injective.} Finally, a constructive version of the mentioned constructions remains an open problem.

\clearpage
\appendix
\section{Excursion: What is Natural About a Natural Universe Hierarchy?}\label{sec:what-is-natural-about-natural-hierarchy}

We have alluded to the categorical concepts of \emph{functoriality} and \emph{naturality} in our design of the so-called ``natural universe hierarchy'' above. We now substantiate this in a technical sense, working internal to a conservative higher-order extension of the logical framework of second-order algebraic theories.\footnote{See \citet{gratzer-sterling-2020} for details of conservativity of the higher-order extension.}

\begin{definition}
  Let $B:\SORT$ be a sort and let $E:B\to\SORT$ be a family of sorts. Then for any category $\mathbb{X}$ we may define a category
  $
    \Fam_E\prn{\mathbb{X}}
  $
  as follows:
  \begin{enumerate}
    \item An object of $\Fam_E\prn{\mathbb{X}}$ is a pair $(b,x)$ where $b:B$ and $x\colon E\prn{b}\to \mathbb{X}$.
    \item A morphism from $(b,x)$ to $(b',x')$ is given by a pair $(f,h)$ where of a function $f\colon E\prn{b}\to E\prn{b'}$ together with a family of morphisms $h:x\to (x'\circ f) \colon \mathbb{X}^{E\prn{b}}$.
  \end{enumerate}
\end{definition}

\begin{definition}
  Let $\mathbb{X}$ be a category; we define the meta-category $\FAM(\mathbb{X})$ to consist of pairs $(I,x)$ where $I$ is a sort and $x\colon I\to\mathbb{X}$ is a family of objects of $\mathbb{X}$. An arrow $(I,x) \to (I', x')$ is given by a function $f\colon I\to I'$ together with an arrow $h \colon x \to (x'\circ f)$ in $\mathbb{X}^I$.
\end{definition}

There is an evident inclusion $\Fam_E(\mathbb{X}) \hookrightarrow \FAM(\mathbb{X})$ sending a pair $\prn{b,x}$ to $\prn{E\prn{b}, x}$.

\begin{construction}[The functorial hierarchy]
  If we write $\CatLvl$ for the preorder of levels $i:\SortLvl$ under the inequality relation $i\leq j$, then the assignment of types $\TpUniv_i:\SortTp$ and lifting operations operations $\Lift{i}{j} : \SortTm\prn{\TpUniv_i}\to \SortTm\prn{\TpUniv_j}$ satisfying the various equational laws for $\Lift{i}{i}$ and $\Lift{j}{k}\circ\Lift{i}{j}$ and $\TpEl_j\circ \Lift{i}{j}$ can be bundled into a single functor
  \[
    (\TpUniv_\bullet, \TpEl_\bullet) \colon \CatLvl\to \Fam_\SortTm\prn{\SortTp}\text{,}
  \]
  where we regard $\SortTp$ as a discrete category.
\end{construction}

\begin{construction}[Codes for the universes]
  We now consider the \emph{coproduct} of all the co-slice preorders $i^+/\CatLvl$, which we shall write as $+/\CatLvl$. This is equivalently the \emph{inverted Artin glueing} of the functor $(-)^+\colon \SortLvl\to\CatLvl$ in which $\SortLvl$ is taken to be a discrete preorder:
  \[
    \begin{tikzcd}
      +/\CatLvl
      \ar[r, "\Cod_+"] \ar[d,"\Dom_+"']\ar[dr, phantom, "\Uparrow"] & \CatLvl \ar[d,equals,nfold]\\
      \SortLvl \ar[r, "{(-)^+}"'] & \CatLvl
    \end{tikzcd}
    \quad
    \text{or equivalently}
    \quad
    \begin{tikzcd}
      & \CatLvl
      \\
      |[elbow nw]|+/\CatLvl \ar[r]\ar[ur, sloped,"\Cod_+"] \ar[d, "\Dom_+"'] & \CatLvl^\to \ar[d, "\Dom"]\ar[u, "\Cod"']
      \\
      \SortLvl \ar[r, "{(-)^+}"'] & \CatLvl
    \end{tikzcd}
  \]

  An object of $+/\CatLvl$ is a pair $(i,j)$ with $i^+\leq j$; we have an inequality $(i,j)\leq (i',j')$ in $+/\CatLvl$ if and only if $i = i'$ and $j \leq j'$.
  Now the rules for forming $\CodeUniv_i^j$ and the equations governing $\TpEl_j\prn[\big]{\CodeUniv_i^j}$ and $\Lift{j}{k}\CodeUniv_i^j$ can be expressed most parsimoniously as a single natural transformation:
  \[
    \begin{tikzcd}[column sep=huge]
      +/\CatLvl
        \ar[r, "\Dom_+"]
        \ar[d, "\Cod_+"']
        \ar[dr,phantom, "{\Downarrow}\CodeUniv_\circ^\bullet"]
      & \SortLvl
        \ar[d, "\prn{\mathbf{1}, \TpUniv_\circ}"]
      \\
      \CatLvl
        \ar[r, "\prn{\SortTm\prn{\TpUniv_\bullet}, \TpEl_\bullet}"']
      & \FAM(\SortTp)
    \end{tikzcd}
  \]

  Indeed, given $(i,j)$ such that $i^+\leq j$ the component $\CodeUniv_i^j\colon \prn{\mathbf{1},\TpUniv_i} \to \prn{\SortTm\prn{\TpUniv_j}, \TpEl_j}$ amounts to a commuting triangle as follows:
  \[
    \begin{tikzcd}
      \mathbf{1}
        \ar[r, densely dashed, "\CodeUniv_i^j"]
        \ar[dr, sloped, "\TpUniv_i"']
      & \SortTm\prn{\TpUniv_j}
        \ar[d,"\TpEl_j"]
      \\
      &\SortTp
    \end{tikzcd}
  \]

  Naturality of $\CodeUniv_\circ^\bullet$ means precisely that for any $i,j,k$ such that $i^+\leq j \leq k$, we have $\Lift{j}{k}\CodeUniv_i^j = \CodeUniv_i^k$.
\end{construction}

\begin{construction}[Codes for dependent function types]
  The rules for forming the dependent function type codes $\CodePi_i\prn{M,N}$ and the equations governing $\TpEl_i\prn{\CodePi_i\prn{M,N}}$ and $\Lift{i}{j}\prn{\CodePi_i\prn{M,N}}$ can also be bundled into a single natural transformation:
  \[
    \begin{tikzcd}[column sep=large, row sep=small]
      \CatLvl
      \ar[d, equals, nfold]
      \ar[drrrrr, phantom, "{\Downarrow} \CodePi_\bullet"]
      \ar[rrrrr, "{
        \prn[\big]{
          \sum_{\prn{M:\SortTm\prn{\TpUniv_\bullet}}}
          \SortTm\prn{\TpUniv_\bullet}^{\SortTm\prn{\TpEl_\bullet(x)}}
          ,\
          \lambda\prn{M,N}\mathpunct{.}\Pi\prn{\TpEl_\bullet\prn{M}, \TpEl_\bullet\circ N}
        }
      }"]
      &&&&&
        \smash{\FAM\prn{\SortTp}}\vphantom{\CatLvl}
      \ar[d, equals, nfold]
      \\
      \CatLvl
        \ar[rrrrr, "\prn{\SortTm\prn{\TpUniv_\bullet},\TpEl_\bullet}"']
      &&&&& \smash{\FAM\prn{\SortTp}}\vphantom{\CatLvl}
    \end{tikzcd}
  \]

  In each component, the natural transformation above gives us a family of maps \[ \CodePi_i\colon \prn[\big]{\textstyle\sum_{\prn{M:\SortTm\prn{\TpUniv_\bullet}}}
  \SortTm\prn{\TpUniv_\bullet}^{\SortTm\prn{\TpEl_\bullet(x)}}} \to \SortTm\prn{\TpUniv_i}\] satisfying the equation
  $
    \TpEl_i\prn{
      \CodePi_i\prn{M,N}
    }=
    \TpPi\prn{
      \TpEl_i\prn{M},
      \lambda x \mathpunct{.}
      \TpEl_i\prn{N\prn{x}}
    }
  $. The naturality law states precisely that we have
  $
    \Lift{i}{j}\prn{
      \CodePi_i\prn{M,N}
    }
    =
    \CodePi_j\prn[\big] {
      \Lift{i}{j}{M},
      x\mapsto
      \Lift{i}{j}{N\prn{x}}
    }
  $
  for all $i\leq j$.
\end{construction}

\section{Metatheory of \texorpdfstring{\ThyMLTTNat}{NatHier}}

\subsection{Displayed Higher-Order Models Over First-Order Models}

Let $\mathbb{T}$ be a SOGAT and let $\mathcal{M}$ be a first-order model of $\mathbb{T}$. \citet{BocquetPhD} defined a \emph{displayed higher-order model} $\widetilde{\mathcal{M}}$ over $\mathcal{M}$ to be a higher-order model of $\mathbb{T}$ in the \emph{scone}\footnote{The \emph{scone}, or Sierpinski cone, of a category $\mathcal{C}$ is the gluing along the global section functor $\Gamma_\mathcal{C} = \mathcal{C}(1_\mathcal{C},-)$, which is explicitely the comma category $(\Set \downarrow \Gamma_\mathcal{C})$.} of $\Psh\prn{\mathbb{C}_{\mathcal{M}}}$ that restricts to $\mathcal{M}$ under the codomain functor:
\[
  \begin{tikzcd}
    \mathbb{T}
      \ar[dr,sloped,"\mathcal{M}"']
      \ar[r,densely dashed,"\widetilde{\mathcal{M}}"]
    &
    \mathbf{Scn}\prn{\Psh\prn{\mathbb{C}_{\mathcal{M}}}}
      \ar[d, "\mathsf{cod}"]
      \ar[r, "\mathsf{dom}"]
      \ar[dr,phantom,"\Swarrow"]
    &
    \Set
      \ar[d, equals, nfold]
    \\
    &
    \Psh\prn{\mathbb{C}_{\mathcal{M}}}
    \ar[r, "\mathsf{eval}_{\mathbf{1}}"']
    &
    \Set
  \end{tikzcd}
\]

In type theoretic terms, this amounts to defining for each sort in $\mathbb{T}$ a family of sets indexed in the \emph{closed} or \emph{global} elements of $\mathcal{M}.X$, and so on. In the case that $\mathbb{T}$ is the SOGAT generated by a single representable family $A:\SortTp\vdash \SortTm\prn{A}:\REPSORT$, then a displayed higher-order model over a first-order model $\mathcal{M}$ would amount to the following data:

\iblock{
  \mrow{\widetilde{\mathcal{M}}.\SortTp\colon \mathcal{M}.\SortTp\prn{\mathbf{1}}\to\Set}
  \mrow{
    \widetilde{\mathcal{M}}.\SortTm\colon\prn{
      A:\mathcal{M}.\SortTp\prn{\mathbf{1}},
      A^\bullet : \widetilde{\mathcal{M}}.\SortTp\prn{A},
      M:\mathcal{M}.\SortTm\prn{\mathbf{1},A}
    }\to \Set
  }
}

\subsubsection{From Displayed Higher-Order Models to Displayed First-Order Models}

Now, a displayed higher-order model $\widetilde{\mathcal{M}}$ may be ``recontextualised'' to a first-order model $\mathsf{Cxl}\prn{\widetilde{\mathcal{M}}}$ displayed over $\mathcal{M}$. The category of contexts of $\mathsf{Cxl}\prn{\widetilde{\mathcal{M}}}$ is simply the scone $\mathbf{Scn}\prn{\mathbb{C}_{\mathcal{M}}}$ displayed over $\mathbb{C}_{\mathcal{M}}$ via the codomain functor; in other words, a context of $\mathsf{Cxl}\prn{\widetilde{\mathcal{M}}}$ over $\Gamma\in\mathbb{C}_{\mathcal{M}}$ is given by a family of sets $\Gamma^\bullet$ indexed in the global elements $\mathbb{C}_{\mathcal{M}}\prn{\mathbf{1},\Gamma}$. Then, in our example, a type in context $\prn{\Gamma,\Gamma^\bullet}$ over $A\in\mathcal{M}.\SortTp\prn{\Gamma}$ is given by a family of sets
\[
  A^\bullet \colon
  \prn{
    \gamma : \mathbb{C}_{\mathcal{M}}\prn{\mathbf{1},\Gamma},
    \gamma^\bullet : \Gamma^\bullet\prn{\gamma}
  }\to
  \widetilde{\mathcal{M}}.\SortTp\prn{A[\gamma]}\text{,}
\]
and a term of type $\prn{A,A^\bullet}$ displayed over $M :\mathcal{M}.\SortTm\prn{\Gamma,A}$ is given by a section
\[
  M^\bullet \colon
  \prn{
    \gamma : \mathbb{C}_{\mathcal{M}}\prn{\mathbf{1},\Gamma},
    \gamma^\bullet : \Gamma^\bullet\prn{\gamma}
  }\to
  \widetilde{\mathcal{M}}.\SortTm\prn{A[\gamma],A^\bullet\prn{\gamma,\gamma^\bullet},M[\gamma]}
\]
and so on.

\subsubsection{Sections of Displayed Higher-Order Models}
A \emph{section} of a displayed higher-order model $\widetilde{\mathcal{M}}$ over a first-order model $\mathcal{M}$ is defined to be a section (in the first-order sense) of the contextualisation $\mathsf{Cxl}\prn{\widetilde{\mathcal{M}}}$ over $\mathcal{M}$. By unraveling, we see that this amounts to a section $\bbrk{-}\colon \mathbb{C}_{\mathcal{M}}\to \mathbf{Scn}\prn{\mathbb{C}_{\mathcal{M}}}$ of the codomain projection together with the following assignments:

\begin{mathparsmall}
  \ebrule{
    \hypo{A:\mathcal{M}.\SortTp\prn{\Gamma}}
    \hypo{\gamma:\mathbb{C}_{\mathcal{M}}\prn{\mathbf{1},\Gamma}}
    \hypo{\gamma^\bullet:\bbrk{\Gamma}\gamma}
    \infer3{
      \bbrk{A}\,\gamma\,\gamma^\bullet : \widetilde{\mathcal{M}}.\SortTp\prn{A[\gamma]}
    }
  }
  \and
  \ebrule{
    \hypo{A:\mathcal{M}.\SortTp\prn{\Gamma}}
    \hypo{\gamma:\mathbb{C}_{\mathcal{M}}\prn{\mathbf{1},\Gamma}}
    \hypo{\gamma^\bullet:\bbrk{\Gamma}\gamma}
    \hypo{M: \mathcal{M}.\SortTm\prn{\Gamma,A}}
    \infer4{
      \bbrk{M}\,\gamma\,\gamma^\bullet :
      \widetilde{\mathcal{M}}.\SortTm\prn{A[\gamma],\bbrk{A}\,\gamma\,\gamma^\bullet, M\brk{\gamma}}
    }
  }
\end{mathparsmall}

\emph{Mutatis mutandis}, the scheme above carries on to any second-order generalised algebraic theory.

\subsection{The Relative Induction Principle for Renamings}

Let $\mathcal{M}$ be a displayed first-order model of $\mathbb{T}$ over the variable parameterisation $\VarParam{\ModS}$. A \emph{variable structure} on $\mathcal{M}$ is given by the following data:
\[
  \ebrule{
    \hypo{A:\VarParam{\ModS}.\SortTp\prn{\mathbf{1}}}
    \hypo{A^\bullet : \mathcal{M}.\SortTp\prn{\mathbf{1},A}}
    \hypo{x:\mathsf{Var}_{\ModS}\prn{A}}
    \infer3{
      \mathcal{M}.\mathsf{var}\prn{A,A^\bullet,x} : \mathcal{M}.\SortTm\prn{\mathbf{1},A,A^\bullet,\mathsf{var}\prn{A,x}}
    }
  }
\]
Similarily, a \emph{level structure} on $\mathcal{M}$ is given by the following data:
\[
  \ebrule{
    \hypo{\alpha : \LvlVar_{\ModS}}
    \infer1{
      \mathcal{M}.\lvlVar\prn{\alpha} : \mathcal{M}.\SortLvl\prn{\lvlVar\prn{i}}
    }
  }
\]

We adapt Bocquet's main result~\cite[Thm.~6.3.10]{BocquetPhD}, which he dubs the \emph{relative induction principle for renamings} to our setting with level variables.

\begin{theorem}\label{thm:relative-induction-principle}
  When $\mathcal{S}$ is the initial first-order $\mathbb{T}$-model, then $\VarParam{\ModS}$ (trivially displayed over itself) is the \emph{initial} displayed first-order model over $\VarParam{\ModS}$ with a variable and level structure. In other words, if $\mathcal{M}$ is a displayed first-order model of $\VarParam{\ModS}$ equipped with a variable and level structure, then there is a universal section $\bbrk{-}$ of $\mathcal{M}$ satisfying
  \[
    \bbrk{\mathsf{var}\prn{A,x}}
    =
    \mathcal{M}.\mathsf{var}\prn{
      A,
      \bbrk{A},
      x
    }
    \qquad
    \text{and}
    \qquad
    \bbrk{\lvlVar\prn{\alpha}}
    =
    \mathcal{M}.\lvlVar\prn{
      \alpha
    }\text{.}
  \]
\end{theorem}

\subsection{The Normalisation Model of \texorpdfstring{\ThyMLTTNat}{NatHier}}\label{sec:normalisation-model}

We now define a displayed higher-order model $\widetilde{\ModS}$ over the variable parameterisation $\VarParam{\ModS}$, showing only a representative fragment of the full definition. Given $A:\SortTp\prn{\mathbf{1}}$, a displayed type $A^\bullet : \widetilde{\ModS}.\SortTp\prn{A}$ is given by the following data:
\begin{enumerate}
  \item a \emph{normal form} $A^\bullet\!.\mathsf{nf} : \IsNfTp{A}$;
  \item for each $M:\SortTm\prn{\mathbf{1},A}$, a set $A^\bullet\prn{M}$ of \emph{computability witnesses} equipped
   with functions
   \[
      \IsNeTm{A}{M} \xrightarrow[\textit{unquote}]{A^\bullet\!.\mathsf{u}_M} A^\bullet\prn{M}\xrightarrow[\textit{quote}]{A^\bullet\!.\mathsf{q}_M} \IsNfTm{A}{M}\text{.}
   \]
\end{enumerate}

Given $M:\SortTm\prn{\mathbf{1},A}$, a displayed term $M^\bullet : \widetilde{\ModS}.\SortTm\prn{A,A^\bullet,M}$ is given by an element of $A^\bullet\prn{M}$. The displayed interpretations of the basic type constructors aside from universes are standard (see \emph{e.g.}\ \citet{sconing}, so we will focus solely on universes.

\begin{construction}[Universe levels]
    We have a projection map
    $
      p: \prn{\mathsf{LvlVar} + \mathbf{1}} \times \mathbb{N} \longrightarrow \mathcal{S}.\SortLvl
    $
    sending $(\alpha,n)$ to $\alpha^{+^n}$ and $(*,n)$ to $\bar{n}$.
    Given a universe level $i:\SortLvl\prn{\mathbf{1}}$, a \emph{displayed universe level} over $i$ is defined to be the fibre of this projection map over $i$:
    \[
      \widetilde{\mathcal{S}}.\SortLvl\prn{i} \coloneq
      p^{-1}\prn{i}
      =
      \brc[\big]{\prn{\alpha^?,n} \bigm\vert p\prn{\alpha^?,n} = i}
    \]

    We lift the basic level operations to computability witnesses as follows:
    \[
      \widetilde{\ModS}.0 \coloneq \prn{*,0}
      \qquad
      \widetilde{\ModS}.\prn{-}^+\prn{i,\prn{\alpha^?,n}} \coloneq \prn{\alpha^?,n+1}
    \]

\end{construction}

Next we shall define the displayed universes $\widetilde{\ModS}.\TpUniv_{\prn{i,n}} : \widetilde{\ModS}.\SortTp\prn{\TpUniv_i}$ for each $n:\widetilde{\ModS}.\SortLvl\prn{i}$.

\begin{construction}[Computability witnesses for type codes]
  Given a code $M:\SortTm\prn{\mathbf{1},\TpUniv_i}$, we define a \emph{computability witness} $M^\bullet : \widetilde{\ModS}.\TpUniv_{\prn{i,(\alpha,n)}}\prn{M}$ to consist of the following data:
  \begin{enumerate}
    \item for each $m\geq n$, a normal form $M^\bullet\!.\mathsf{nf}_m : \IsNfTm{\TpUniv_{\bar{m}}}{\Lift{i}{\bar{m}}\prn{M}}$;
    \item a decoded normal form $M^\bullet\!.\mathsf{dnf} : \IsNfTp{\TpEl_i\prn{M}}$;
    \item for each $N:\SortTm\prn{\mathbf{1},\TpEl_i\prn{M}}$, a $\mathcal{U}_n$-small set $M^\bullet\prn{N}$ of computability witnesses equipped with functions
    \[
      \IsNeTm{\TpEl_i\prn{M}}{N} \xrightarrow[\textit{unquote}]{M^\bullet\!.\mathsf{u}_N} M^\bullet\prn{N}\xrightarrow[\textit{quote}]{M^\bullet\!.\mathsf{q}_N} \IsNfTm{\TpEl_i\prn{M}}{N}\text{.}
    \]
  \end{enumerate}
\end{construction}

\begin{construction}[(Un)quotation in the universe]
  We shall now define the unquotation map
  \[
    \IsNeTm{\TpUniv_i}{M}\xrightarrow[\textit{unquote}]{\widetilde{\ModS}.\TpUniv_{\prn{i,n}}.\mathsf{u}_M} \widetilde{\ModS}.\TpUniv_{\prn{i,n}}\prn{M}
  \]
  for each universe so that the computability witnesses of elements of a neutral type are given by neutral elements of that neutral type.

  \iblock{
    \mrow{
      \widetilde{\ModS}.\TpUniv_{\prn{i,n}}.\mathsf{u}_M\prn{\mathfrak{m}} : \widetilde{\ModS}.\TpUniv_{\prn{i,n}}\prn{M}
    }
    \mrow{
      \widetilde{\ModS}.\TpUniv_{\prn{i,n}}.\mathsf{u}_M\prn{\mathfrak{m}}.\mathsf{nf}_m \coloneq
      \textsc{nf-lift}_m\prn{\mathfrak{m}}
    }
    \mrow{
      \widetilde{\ModS}.\TpUniv_{\prn{i,n}}.\mathsf{u}_M\prn{\mathfrak{m}}.\mathsf{dnf} \coloneq
      \textsc{nftp-el}\prn{\mathfrak{m}}
    }
    \mrow{
      \widetilde{\ModS}.\TpUniv_{\prn{i,n}}.\mathsf{u}_M\prn{\mathfrak{m}}\prn{N} \coloneq \IsNeTm{\TpEl_i\prn{M}}{N}
    }
    \mrow{
      \widetilde{\ModS}.\TpUniv_{\prn{i,n}}.\mathsf{u}_M\prn{\mathfrak{m}}.\mathsf{u}_N\prn{\mathfrak{n}} \coloneq \mathfrak{n}
    }
    \mrow{
      \widetilde{\ModS}.\TpUniv_{\prn{i,n}}.\mathsf{u}_M\prn{\mathfrak{m}}.\mathsf{q}_N\prn{\mathfrak{n}} \coloneq \textsc{nf-ne-shift}\prn{\mathfrak{n}}
    }
  }

  The quotation map
  $
    \widetilde{\ModS}.\TpUniv_{\prn{i,n}}\prn{M}
    \xrightarrow[\textit{quote}]{\widetilde{\ModS}.\TpUniv_{\prn{i,n}}.\mathsf{q}_M}
    \IsNfTm{\TpUniv_i}{M}
  $
  is defined by projection:

  \iblock{
    \mrow{
      \widetilde{\ModS}.\TpUniv_{\prn{i,n}}.\mathsf{q}_M\prn{M^\bullet} : \IsNfTm{\TpUniv_i}{M}
    }
    \mrow{
      \widetilde{\ModS}.\TpUniv_{\prn{i,n}}.\mathsf{q}_M\prn{M^\bullet} \coloneq M^\bullet\!.\mathsf{nf}_n
    }
  }
\end{construction}

\begin{construction}[Universe lifting and decoding]
  We next define the displayed interpretation of the universe lifting and decoding maps. Here we are taking advantage of the fact that in any presheaf topos, the Hofmann--Streicher universes satisfy the laws specified in \ThyMLTTNat.

  \iblock{
    \mrow{
      \widetilde{\ModS}.\Lift{\prn{i,m}}{\prn{j,n}}\prn{M,M^\bullet} : \widetilde{\ModS}.\TpUniv_{\prn{j,n}}\prn[\big]{\Lift{i}{j}{M}}
    }
    \mrow{
      \widetilde{\ModS}.\Lift{\prn{i,m}}{\prn{j,n}}\prn{M,M^\bullet}.\mathsf{nf}_o
      \coloneq
      M^\bullet\!.\mathsf{nf}_o
    }
    \mrow{
      \widetilde{\ModS}.\Lift{\prn{i,m}}{\prn{j,n}}\prn{M,M^\bullet}.\mathsf{dnf}
      \coloneq
      M^\bullet\!.\mathsf{dnf}
    }
    \mrow{
      \widetilde{\ModS}.\Lift{\prn{i,m}}{\prn{j,n}}\prn{M,M^\bullet}\prn{N}
      \coloneq
      \Lift{m}{n}{
        M^\bullet\prn{N}
      }
    }
    \mrow{
      \widetilde{\ModS}.\Lift{\prn{i,m}}{\prn{j,n}}\prn{M,M^\bullet}.\mathsf{u}_N\prn{\mathfrak{n}}
      \coloneq
      M^\bullet\!.\mathsf{u}_N\prn{\mathfrak{n}}
    }
    \mrow{
      \widetilde{\ModS}.\Lift{\prn{i,m}}{\prn{j,n}}\prn{M,M^\bullet}.\mathsf{q}_N\prn{N^\bullet}
      \coloneq
      M^\bullet\!.\mathsf{q}_N\prn{N^\bullet}
    }
  }

  The decoding map is defined similarly:

  \iblock{
    \mrow{
      \widetilde{\ModS}.\TpEl_{\prn{i,n}}\prn{M,M^\bullet} : \widetilde{\ModS}.\SortTp\prn{M}
    }
    \mrow{
      \widetilde{\ModS}.\TpEl_{\prn{i,n}}\prn{M,M^\bullet}.\mathsf{nf} \coloneq M^\bullet.\mathsf{dnf}
    }
    \mrow{
      \widetilde{\ModS}.\TpEl_{\prn{i,n}}\prn{M,M^\bullet}\prn{N} \coloneq \TpEl_n\prn{
        M^\bullet\prn{N}
      }
    }
    \mrow{
      \widetilde{\ModS}.\TpEl_{\prn{i,n}}\prn{M,M^\bullet}.\mathsf{u}_N\prn{\mathfrak{n}} \coloneq
      M^\bullet\!.\mathsf{u}_N\prn{\mathfrak{n}}
    }
    \mrow{
      \widetilde{\ModS}.\TpEl_{\prn{i,n}}\prn{M,M^\bullet}.\mathsf{q}_N\prn{N^\bullet} \coloneq
      M^\bullet\!.\mathsf{q}_N\prn{N^\bullet}
    }
  }
\end{construction}

The displayed interpretation of level products and abstraction is a direct adaptation of the usual method for function types, which we omit for space. 

\begin{construction}[Variable and level structure]\label{con:var-and-lvl-str}
  Finally we exhibit variable and level structure.

  \iblock{
    \mrow{
      \widetilde{\ModS}.\var\prn{A,A^\bullet,x} \coloneq A^\bullet.\mathsf{u}_{\var\prn{A,x}}~\prn{
        \textsc{ne-var} : \IsNeTm{A}{\var\prn{A,x}}
      }
    }
    \mrow{
      \widetilde{\ModS}.\lvlVar\prn{\alpha} \coloneq
      \prn{\alpha, 0}
    }
  }
\end{construction}

\subsection{Injectivity of Decoding for \texorpdfstring{\ThyMLTTNat}{NatHier}}

\ThmInjDecoding*
\begin{proof}
  By \zcref{thm:normalisation} we may generalise the goal over normal forms of $M$ and $N$ to show:

  \iblock{
    \mrow{
      M,N:\VarParam{\ModS}.\SortTm\prn{\mathbf{1},\TpUniv_i}
      \mid
      \IsNfTm{\TpUniv_i}{M}, \IsNfTm{\TpUniv_i}{N},
      \TpEl_i\prn{M} = \TpEl_i\prn{N}
      \vdash
      M = N
    }
  }

  We now proceed by induction on normal forms; cross-cases are ruled out by \zcref{cor:no-confusion}. The other sort of case we must consider is when the two normal forms have the same head constructor. For example, in the case of codes for $\Pi$-types:

  \iblock{
    \mrow{
      \IsNfTm{\TpUniv_i}{M},
      \IsNfTm{\TpUniv_i}{M'},
      \forall x\mathpunct{.}\IsNfTm{\TpUniv_i}{N\brk{\var\prn{x}}},
      \forall x\mathpunct{.}\IsNfTm{\TpUniv_i}{N'\brk{\var\prn{x}}},
    }
    \mrow{
      \TpPi\prn{\TpEl_i\prn{M},\TpEl_i\prn{N}} =
      \TpPi\prn{\TpEl_i\prn{M'},\TpEl_i\prn{N'}}
    }
    \iblock{
      \mrow{
        \vdash
        \CodePi_i\prn{M,N} = \CodePi_i\prn{M',N'}
      }
    }
  }

  By \zcref{cor:injectivity-of-type-constructors}, we may assume
  $\TpEl_i\prn{M} = \TpEl_i\prn{M'}$ and $\TpEl_i\prn{N}=\TpEl_i\prn{N'}$.

  \iblock{
    \mrow{
      \IsNfTm{\TpUniv_i}{M},
      \IsNfTm{\TpUniv_i}{M'},
      \forall x\mathpunct{.}\IsNfTm{\TpUniv_i}{N\brk{\var\prn{x}}},
      \forall x\mathpunct{.}\IsNfTm{\TpUniv_i}{N'\brk{\var\prn{x}}},
    }
    \mrow{
      \TpEl_i\prn{M} = \TpEl_i\prn{M'},
      \TpEl_i\prn{N}=\TpEl_i\prn{N'}
    }
    \iblock{
      \mrow{
        \vdash
        \CodePi_i\prn{M,N} = \CodePi_i\prn{M',N'}
      }
    }
  }

  By our induction hypothesis, we have $M = M'$ and $\forall x\mathpunct{.}N\brk{\var\prn{x}} = N'\brk{\var\prn{x}}$. By \zcref{lemma:substitution-is-iso} we have $N=N'$ so we are done.
\end{proof}

\clearpage

\begin{acks}
We thank Andr\'as Kov\'acs for his feedback and advice; we are grateful also to an anonymous referee for pointing out an error in our elaboration rules. We also thank Thierry Coquand and Peter Dybjer for early discussions regarding the normalisation of the theory with universe polymorphism. This work was funded in part by the United States Air Force Office of Scientific Research under grant FA9550-23-1-0728 (\emph{New Spaces for Denotational Semantics}; Dr Tristan Nguyen, Program Manager). Views and opinions expressed are however those of the authors only and do not necessarily reflect those of AFOSR.
\end{acks}

\section*{Data Availability Statement}

The Haskell implementation described in this paper is available on the first author's GitHub profile,\footnote{\url{https://github.com/raphael-sterbac/elaboration-universes}} and is archived on Zenodo~\citep{sterbac-sterling-artifact}.

\section*{Disclosure on Use of Human Intelligence in Research}

We disclose that human brains \emph{were} used in both the research and the production of this paper. We did not make use of generative AI.
 
\clearpage

\bibliography{refs}

@phdthesis{lennon-bertrand-phd,
  author       = {Meven Lennon-Bertrand},
  title        = {Bidirectional Typing for the Calculus of Inductive Constructions},
  school       = {Nantes Universit{\'e}},
  year         = {2022},
  month        = jun,
  url          = {https://www.meven.ac/documents/22-phd.pdf}
}

@inproceedings{lennon-bertrand-2021,
  author       = {Meven Lennon-Bertrand},
  title        = {Complete Bidirectional Typing for the Calculus of Inductive Constructions},
  year         = {2021},
  doi          = {10.4230/LIPIcs.ITP.2021.24}
}

@InProceedings{hugunin:2020,
  author =	{Hugunin, Jasper},
  title =	{{Why Not W?}},
  booktitle =	{26th International Conference on Types for Proofs and Programs (TYPES 2020)},
  pages =	{8:1--8:9},
  series =	{Leibniz International Proceedings in Informatics (LIPIcs)},
  ISBN =	{978-3-95977-182-5},
  ISSN =	{1868-8969},
  year =	{2021},
  volume =	{188},
  editor =	{de'Liguoro, Ugo and Berardi, Stefano and Altenkirch, Thorsten},
  publisher =	{Schloss Dagstuhl -- Leibniz-Zentrum f{\"u}r Informatik},
  address =	{Dagstuhl, Germany},
  URL =		{https://drops.dagstuhl.de/entities/document/10.4230/LIPIcs.TYPES.2020.8},
  URN =		{urn:nbn:de:0030-drops-138876},
  doi =		{10.4230/LIPIcs.TYPES.2020.8}
}

@unpublished{herbelin-2005,
	title = {Type inference with algebraic universes in the {Calculus} of {Inductive} {Constructions}},
    year={2005},
	author = {Herbelin, Hugo},
    url={https://pauillac.inria.fr/~herbelin/publis/univalgcci.pdf}
}

@misc{int-tree-issues,
    author = {Zakowski, Yannick},
    title = {Universe inconsistency when paired with RelationAlgebra (\#254)},
    note = {Ticket on DeepSpec/InteractionTrees repository related to universe inconsistencies when importing other libraries},
    year = {2023},
    month = mar,
    day = {6},
    url = {https://github.com/DeepSpec/InteractionTrees/issues/254}
}

@article{dfk-2026,
    author = {Danielsson, Nils Anders and Favier, Na\"{\i}m Camille and Kub\'{a}nek, Ond\v{r}ej},
    title = {Normalisation for First-Class Universe Levels},
    year = {2026},
    number = {POPL},
    doi = {10.1145/3776645},
    journal = {Proc. ACM Program. Lang.},
}

@software{foveran,
    title = {Foveran},
    url = {https://github.com/bobatkey/foveran/tree/master},
	author = {Atkey, Bob},
    year = {2016}
}

@software{elabzoo,
    note = {A series of Haskell implementations for elaborating dependently typed languages},
    title = {Elaboration Zoo},
    url = {https://github.com/AndrasKovacs/elaboration-zoo/tree/master},
	author = {Kov\'acs, Andr\'as},
    year = {2016}
}

@software{smalltt,
    title = {\texttt{smalltt}},
    note = {Demo project for several techniques for high-performance elaboration with dependent types},
    url = {https://github.com/AndrasKovacs/smalltt/tree/master},
	author = {Kov\'acs, Andr\'as},
    year = {2021}
}

@article{assaf,
  author       = {Ali Assaf},
  title        = {A calculus of constructions with explicit subtyping},
  journal      = {TYPES},
  year         = {2014},
  url          = {https://hal.science/hal-01097401v2/file/coc-explicit-subtyping-types-2014%20%281%29.pdf},
}

@InProceedings{felicissimo-winterhalter,
  author =	{Felicissimo, Thiago and Winterhalter, Th\'{e}o},
  title =	{{Impredicativity, Cumulativity and Product Covariance in the Logical Framework Dedukti}},
  booktitle =	{9th International Conference on Formal Structures for Computation and Deduction (FSCD 2024)},
  year =	{2024},
  doi =		{10.4230/LIPIcs.FSCD.2024.21},
}

@misc{coquand-sterbac-2026,
	title = {Cumulative hierarchies of universes and their equivalence in dependent type theory},
	howpublished = {Talk at TYPES 2026},
	author = {Coquand, Thierry and Sterbac, Raphaël},
	year         = {2026},
    url = {https://types2026.cse.chalmers.se/abstracts/23.pdf},
}

@misc{bezem-sozeau-2025,
	title = {Algebraic Universes and Variances For All},
	howpublished = {Talk at TYPES},
	author = {Sozeau, Matthieu and Bezem, Marc},
	year         = {2025},
    url = {https://msp.cis.strath.ac.uk/types2025/abstracts/TYPES2025_paper21.pdf},
}

@misc{lynch-2026,
	title = {An Algebraic Approach to the Static/Dynamic Phase Distinction for Module Calculi},
	howpublished = {Talk at TYPES},
	author = {Lynch, Owen},
	year         = {2026},
    url = {https://types2026.cse.chalmers.se/abstracts/34.pdf},
}

@misc{sterling-2025,
	title = {Is it time for a new proof assistant?},
	howpublished = {Talk at the HoTT Electric Seminar},
	author = {Sterling, Jon},
	year         = {2025},
    url = {https://hottest-seminar.github.io/hottestfiles/Sterling-2025-09-25-HoTTEST.pdf},
}

@InProceedings{timany-jacobs-2015,
author="Timany, Amin
and Jacobs, Bart",
editor="Leucker, Martin
and Rueda, Camilo
and Valencia, Frank D.",
title="First Steps Towards Cumulative Inductive Types in CIC",
booktitle="Theoretical Aspects of Computing - ICTAC 2015",
year="2015",
publisher="Springer International Publishing",
address="Cham",
pages="608--617",
isbn="978-3-319-25150-9",
doi = {10.1007/978-3-319-25150-9_36},
}

@article{coquand-type-checking,
title = {An algorithm for type-checking dependent types},
journal = {Science of Computer Programming},
year = {1996},
doi = {https://doi.org/10.1016/0167-6423(95)00021-6},
author = {Thierry Coquand},
}

@misc{kovacs-observing-equality,
	title = {Observing {Definitional} {Equality}},
	howpublished = {A talk given at WITS 2026},
	author = {Kov\'acs, Andr\'as},
	year         = {2026},
}

@inproceedings{fiore:2002,
	author = {Fiore, Marcelo P.},
	address = {Pittsburgh, PA, USA},
	publisher = acm,
	booktitle = {Proceedings of the 4th ACM SIGPLAN International Conference on Principles and Practice of Declarative Programming},
	year = {2002},
	doi = {10.1145/571157.571161},
	isbn = {1-58113-528-9},
	pages = {26--37},
	series = {PPDP '02},
	title = {Semantic Analysis of Normalisation by Evaluation for Typed Lambda Calculus},
}

@article{adaptt,
    author = {Adjedj, Arthur and Lennon-Bertrand, Meven and Benjamin, Thibaut and Maillard, Kenji},
    title = {AdapTT: Functoriality for Dependent Type Casts},
    year = {2026},
    number = {POPL},
    doi = {10.1145/3776664},
    journal = {Proc. ACM Program. Lang.},
}

@misc{BocquetPhD,
  author =	{Bocquet, Rafa\"{e}l},
  title =	{{Relative induction principles for second-order generalized algebraic theories}},
  year =	{2026},
  howpublished = {PhD thesis},
  url = {https://rafaelbocquet.gitlab.io/pdfs/thesis.pdf}
}

@InProceedings{sconing,
  author =	{Bocquet, Rafa\"{e}l and Kaposi, Ambrus and Sattler, Christian},
  title =	{{For the Metatheory of Type Theory, Internal Sconing Is Enough}},
  year =	{2023},
  doi =		{10.4230/LIPIcs.FSCD.2023.18},
}

@article{mclarty-2020, title={The large structures of {Grothendieck} founded on finite-order arithmetic}, volume={13}, DOI={10.1017/S1755020319000340}, number={2}, journal={The Review of Symbolic Logic}, author={McLarty, Colin}, year={2020}, pages={296–325}}

@unpublished{shulman-2019,
	author = {Shulman, Michael},
	year = {2019},
	month = apr,
	eprint = {1904.07004},
	eprinttype = {arXiv},
	note = {Unpublished manuscript},
	title = {All $(\infty,1)$-toposes have strict univalent universes},
  doi = {10.48550/arXiv.1904.07004}
}

@book{martin-lof-1984,
	author = {Martin-L\"{o}f, Per},
	publisher = {Bibliopolis},
	year = {1984},
	isbn = {88-7088-105-9},
	pages = {iv+91},
	series = {Studies in Proof Theory},
	subtitle = {Notes by Giovanni Sambin},
	title = {Intuitionistic type theory},
	volume = {1},
}

@misc{dybjer-2025-universes,
	author       = {Peter Dybjer},
	title        = {Are Universes Open or Closed?},
	howpublished = {A talk given at \emph{Types and Topology: A workshop in honour of Martín Escardó's 60th birthday}},
	month        = dec,
	year         = {2025},
	note         = {Birmingham, 17--18 December 2025},
	url = {https://tdejong.com/mhe60/slides/dybjer.pdf}
}

@phdthesis{barras-2012,
	author       = {Bruno Barras},
	title        = {Semantical Investigations in Intuitionistic Set Theory and Type Theories with Inductive Families},
	howpublished = {Université Paris 7 --- Denis Diderot},
	year         = {2012},
	type         = {Habilitation à Diriger des Recherches (HDR)},
	url = {https://lsv.ens-paris-saclay.fr/~barras/habilitation/barras-habilitation.pdf}
}

@InProceedings{timany-sozeau-2018,
	author =	{Timany, Amin and Sozeau, Matthieu},
	title =	{{Cumulative Inductive Types In Coq}},
	booktitle =	{3rd International Conference on Formal Structures for Computation and Deduction (FSCD 2018)},
	pages =	{29:1--29:16},
	series =	{Leibniz International Proceedings in Informatics (LIPIcs)},
	ISBN =	{978-3-95977-077-4},
	ISSN =	{1868-8969},
	year =	{2018},
	volume =	{108},
	editor =	{Kirchner, H\'{e}l\`{e}ne},
	publisher =	{Schloss Dagstuhl -- Leibniz-Zentrum f{\"u}r Informatik},
	address =	{Dagstuhl, Germany},
	URL =		{https://drops.dagstuhl.de/entities/document/10.4230/LIPIcs.FSCD.2018.29},
	URN =		{urn:nbn:de:0030-drops-91991},
	doi =		{10.4230/LIPIcs.FSCD.2018.29}
}

@phdthesis{dagand:2013,
	author = {Dagand, Pierre-\'{E}variste},
	address = {Glasgow, Scotland},
	school = {University of Strathclyde},
	year = {2013},
	month = aug,
	title = {A Cosmology of Datatypes: Reusability and Dependent Types},
  doi = {10.48730/zt52-kb39},
}

@inproceedings{levitation, author = {Chapman, James and Dagand, Pierre-\'{E}variste and {Mc~Bride}, Conor Titania and Morris, Peter}, title = {The gentle art of levitation}, year = {2010}, isbn = {9781605587943}, publisher = {Association for Computing Machinery}, address = {New York, NY, USA}, url = {https://doi.org/10.1145/1863543.1863547}, doi = {10.1145/1863543.1863547}, booktitle = {Proceedings of the 15th ACM SIGPLAN International Conference on Functional Programming}, pages = {3–14}, numpages = {12}, location = {Baltimore, Maryland, USA}, series = {ICFP '10} }

@incollection{barendregt-1992,
	author = {Barendregt, Hendrik Pieter},
	booktitle = {Handbook of Logic in Computer Science},
	doi = {10.1093/oso/9780198537618.003.0002},
	isbn = {9780198537618},
	month = {12},
	publisher = {Oxford University Press},
	title = {Lambda Calculi with Types},
	year = {1992},
}

@inproceedings{luo-1997,
	author = {Luo, Zhaohui},
	editor = {van Dalen, Dirk and Bezem, Marc},
	address = {Berlin, Heidelberg},
	publisher = {Springer Berlin Heidelberg},
	booktitle = {Computer Science Logic},
	year = {1997},
	isbn = {978-3-540-69201-0},
	pages = {275--296},
	title = {Coercive subtyping in type theory},
}

@unpublished{coquand-2012-presheaf-model,
	author       = {Thierry Coquand},
	title        = {Presheaf model of type theory},
	note = {Unpublished manuscript},
	year         = {2012},
	url         = {http://www.cse.chalmers.se/~coquand/presheaf.pdf},
}

@article{harper-pollack-1991,
title = {Type checking with universes},
journal = {Theoretical Computer Science},
volume = {89},
number = {1},
pages = {107-136},
year = {1991},
issn = {0304-3975},
doi = {10.1016/0304-3975(90)90108-T},
author = {Robert Harper and Robert Pollack},
}

@InProceedings{sozeau-tabareau-2014,
author="Sozeau, Matthieu
and Tabareau, Nicolas",
editor="Klein, Gerwin
and Gamboa, Ruben",
title="Universe Polymorphism in Coq",
booktitle="Interactive Theorem Proving",
year="2014",
publisher="Springer International Publishing",
address="Cham",
pages="499--514",
isbn="978-3-319-08970-6",
doi = {10.1007/978-3-319-08970-6_32},
}

@InProceedings{bcde-2022,
	author =	{Bezem, Marc and Coquand, Thierry and Dybjer, Peter and Escard\'{o}, Mart{\'\i}n},
	title =	{{Type Theory with Explicit Universe Polymorphism}},
	booktitle =	{28th International Conference on Types for Proofs and Programs (TYPES 2022)},
	pages =	{13:1--13:16},
	series =	{Leibniz International Proceedings in Informatics (LIPIcs)},
	ISBN =	{978-3-95977-285-3},
	ISSN =	{1868-8969},
	year =	{2023},
	volume =	{269},
	editor =	{Kesner, Delia and P\'{e}drot, Pierre-Marie},
	publisher =	{Schloss Dagstuhl -- Leibniz-Zentrum f{\"u}r Informatik},
	address =	{Dagstuhl, Germany},
	URL =		{https://drops.dagstuhl.de/entities/document/10.4230/LIPIcs.TYPES.2022.13},
	URN =		{urn:nbn:de:0030-drops-184564},
	doi =		{10.4230/LIPIcs.TYPES.2022.13}
}

@incollection{feferman-2004,
title = {Typical Ambiguity: Trying to Have Your Cake and Eat It Too},
booktitle = {One Hundred Years of Russell's Paradox},
author = {Solomon Feferman},
editor = {Godehard Link},
publisher = {De Gruyter},
address = {Berlin, New York},
pages = {135--152},
doi = {10.1515/9783110199680.135},
isbn = {9783110199680},
year = {2004},
}

@misc{mcbride-2002-crude,
	author       = {{Mc~Bride}, Conor Titania},
	title        = {Crude but Effective Stratification},
	year         = {2002},
	url = {https://personal.cis.strath.ac.uk/conor.mcbride/Crude.pdf},
}

@article{favonia-angiuli-mullanix-2023,
author = {Hou (Favonia), Kuen-Bang and Angiuli, Carlo and Mullanix, Reed},
title = {An Order-Theoretic Analysis of Universe Polymorphism},
year = {2023},
issue_date = {January 2023},
publisher = {Association for Computing Machinery},
address = {New York, NY, USA},
volume = {7},
number = {POPL},
url = {https://doi.org/10.1145/3571250},
doi = {10.1145/3571250},
journal = {Proc. ACM Program. Lang.},
month = jan,
articleno = {57},
numpages = {27}
}

@inproceedings{bmss-2011,
	author = {Birkedal, Lars and M{\o{}}gelberg, Rasmus Ejlers and Schwinghammer, Jan and St\o{}vring, Kristian},
	address = {Washington, DC, USA},
	publisher = {IEEE Computer Society},
	booktitle = {Proceedings of the 2011 IEEE 26th Annual Symposium on Logic in Computer Science},
	year = {2011},
	doi = {10.1109/LICS.2011.16},
	eprint = {1208.3596},
	eprintclass = {cs.LO},
	eprinttype = {arXiv},
	isbn = {978-0-7695-4412-0},
	pages = {55--64},
	title = {First Steps in Synthetic Guarded Domain Theory: Step-Indexing in the Topos of Trees},
}

@article{cherubini-coquand-hutzler-2024,
	author = {Cherubini, Felix and Coquand, Thierry and Hutzler, Matthias},
	year = {2024},
	doi = {10.1017/S0960129524000239},
	journal = {Mathematical Structures in Computer Science},
	number = {9},
	pages = {1008--1053},
	title = {A foundation for synthetic algebraic geometry},
	volume = {34},
}

@phdthesis{reus-1995,
	author = {Reus, Bernhard},
	address = {M\"{u}nchen},
	school = {Ludwig-Maximilians-Universit\"{a}t M\"{u}nchen},
	year = {1995},
	month = nov,
	title = {Program Verification in Synthetic Domain Theory},
}

@article{grodin-li-harper-2026,
author = {Grodin, Harrison and Li, Runming and Harper, Robert},
title = {Abstraction Functions as Types: Modular Verification of Cost and Behavior in Dependent Type Theory},
year = {2026},
issue_date = {January 2026},
publisher = {Association for Computing Machinery},
address = {New York, NY, USA},
volume = {10},
number = {POPL},
doi = {10.1145/3776673},
journal = {Proc. ACM Program. Lang.},
month = jan,
articleno = {31},
numpages = {28},
}

@misc{kammar-2021-pihoc,
	author       = {Ohad Kammar},
	title        = {Higher-order building blocks for statistical modelling},
	howpublished = {Talk at the Joint PPS-PIHOC-DIAPASoN Workshop},
	month        = feb,
	year         = {2021},
	url = {https://denotational.co.uk/talks/pihoc-2021-building-blocks.pdf},
}

@inproceedings{sterling-angiuli-2021,
	author = {Sterling, Jonathan and Angiuli, Carlo},
	address = {Los Alamitos, CA, USA},
	publisher = {IEEE Computer Society},
	booktitle = {2021 36th Annual ACM/IEEE Symposium on Logic in Computer Science (LICS)},
	year = {2021},
	month = jul,
	doi = {10.1109/LICS52264.2021.9470719},
	eprint = {2101.11479},
	eprintclass = {cs.LO},
	eprinttype = {arXiv},
	pages = {1--15},
	title = {Normalization for Cubical Type Theory},
}

@phdthesis{sterling-2021-thesis,
	author = {Sterling, Jonathan},
	school = {Carnegie Mellon University},
	year = {2021},
	doi = {10.5281/zenodo.5709837},
	note = {Version 1.2, revised May 2024},
	number = {CMU-CS-21-142},
	title = {First Steps in Synthetic {Tait} Computability: The Objective Metatheory of Cubical Type Theory},
}

@inproceedings{altenkirch-kaposi-2016-nbe,
	author = {Altenkirch, Thorsten and Kaposi, Ambrus},
	editor = {Kesner, Delia and Pientka, Brigitte},
	address = {Dagstuhl, Germany},
	publisher = {Schloss Dagstuhl--Leibniz-Zentrum fuer Informatik},
	url = {http://drops.dagstuhl.de/opus/volltexte/2016/5972},
	booktitle = {1st International Conference on Formal Structures for Computation and Deduction (FSCD 2016)},
	year = {2016},
	doi = {10.4230/LIPIcs.FSCD.2016.6},
	isbn = {978-3-95977-010-1},
	issn = {1868-8969},
	pages = {6:1--6:16},
	series = {Leibniz International Proceedings in Informatics (LIPIcs)},
	title = {{Normalisation by Evaluation for Dependent Types}},
	volume = {52},
}

@incollection{dybjer-1996,
	author = {Dybjer, Peter},
	editor = {Berardi, Stefano and Coppo, Mario},
	address = {Berlin, Heidelberg},
	publisher = {Springer Berlin Heidelberg},
	booktitle = {Types for Proofs and Programs: International Workshop, TYPES '95 Torino, Italy, June 5--8, 1995 Selected Papers},
	year = {1996},
	isbn = {978-3-540-70722-6},
	pages = {120--134},
	title = {Internal type theory},
}

@phdthesis{uemura-2021-thesis,
	author = {Uemura, Taichi},
	address = {Amsterdam},
	publisher = {Institute for Logic, Language and Computation},
	school = {Universiteit van Amsterdam},
	url = {https://www.illc.uva.nl/cms/Research/Publications/Dissertations/DS-2021-09.text.pdf},
	year = {2021},
	title = {Abstract and Concrete Type Theories},
}

@article{bcde-2021, title={On generalized algebraic theories and categories with families}, volume={31}, DOI={10.1017/S0960129521000268}, number={9}, journal={Mathematical Structures in Computer Science}, author={Bezem, Marc and Coquand, Thierry and Dybjer, Peter and Escardó, Martín}, year={2021}, pages={1006–1023}}

@inproceedings{allen-1987-lics,
	author = {Allen, Stuart Frazier},
	address = {Ithaca, NY},
	publisher = {IEEE Computer Society},
	booktitle = {Proceedings of the Symposium on Logic in Computer Science (LICS '87)},
	year = {1987},
	pages = {215--221},
	title = {A Non-Type-Theoretic Definition of {Martin-L\"{o}f's} Types},
}

@incollection{martin-lof-1994,
	author = {Martin-L\"{o}f, Per},
	editor = {Parrini, Paolo},
	address = {Dordrecht},
	publisher = {Springer Netherlands},
	booktitle = {Kant and Contemporary Epistemology},
	year = {1994},
	isbn = {978-94-011-0834-8},
	pages = {87--99},
	title = {Analytic and Synthetic Judgements in Type Theory},
  doi = {10.1007/978-94-011-0834-8_5},
}

@article{gratzer-kavvos-birkedal-2021,
		title      = {Multimodal Dependent Type Theory},
		author     = {Daniel Gratzer and G. A. Kavvos and Andreas Nuyts and Lars Birkedal},
		url        = {https://lmcs.episciences.org/7571},
		doi        = {10.46298/lmcs-17(3:11)2021},
		journal    = {Logical Methods in Computer Science},
		issn       = {1860-5974},
		volume     = {Volume 17, Issue 3},
		eid        = 11,
		year       = {2021},
		month      = {Jul},
}

@article{gratzer-shulman-sterling-2026,
		title      = {Strict universes for Grothendieck topoi},
	    author = {Gratzer, Daniel and Shulman, Michael and Sterling, Jonathan},
		url        = {http://www.tac.mta.ca/tac/volumes/45/30/45-30.pdf},
		journal    = {Theory and Applications of Categories},
		year       = {2026},
}

@inproceedings{allais-mcbride-boutillier-2013,
author = {Allais, Guillaume and {Mc~Bride}, Conor Titania and Boutillier, Pierre},
title = {New equations for neutral terms: a sound and complete decision procedure, formalized},
year = {2013},
isbn = {9781450323840},
publisher = {Association for Computing Machinery},
address = {New York, NY, USA},
doi = {10.1145/2502409.2502411},
booktitle = {Proceedings of the 2013 ACM SIGPLAN Workshop on Dependently-Typed Programming},
pages = {13–24},
numpages = {12},
location = {Boston, Massachusetts, USA},
series = {DTP '13}
}

@phdthesis{abel-2013,
	author = {Abel, Andreas},
	school = {Ludwig-Maximilians-Universit\"{a}t M\"{u}nchen},
	url = {http://www2.tcs.ifi.lmu.de/~abel/habil.pdf},
	year = {2013},
	title = {Normalization by Evaluation: Dependent Types and Impredicativity},
	type = {Habilitation},
}

@InProceedings{kaposi-xie-2024,
	author =	{Kaposi, Ambrus and Xie, Szumi},
	title =	{{Second-Order Generalised Algebraic Theories: Signatures and First-Order Semantics}},
	booktitle =	{9th International Conference on Formal Structures for Computation and Deduction (FSCD 2024)},
	pages =	{10:1--10:24},
	series =	{Leibniz International Proceedings in Informatics (LIPIcs)},
	ISBN =	{978-3-95977-323-2},
	ISSN =	{1868-8969},
	year =	{2024},
	volume =	{299},
	editor =	{Rehof, Jakob},
	publisher =	{Schloss Dagstuhl -- Leibniz-Zentrum f{\"u}r Informatik},
	address =	{Dagstuhl, Germany},
	URL =		{https://drops.dagstuhl.de/entities/document/10.4230/LIPIcs.FSCD.2024.10},
	URN =		{urn:nbn:de:0030-drops-203396},
	doi =		{10.4230/LIPIcs.FSCD.2024.10}
}

@unpublished{gratzer-sterling-2020,
	author = {Gratzer, Daniel and Sterling, Jonathan},
	year = {2020},
	eprint = {2012.10783},
	eprintclass = {cs.LO},
	eprinttype = {arXiv},
	note = {Unpublished manuscript},
	title = {Syntactic categories for dependent type theory: sketching and adequacy},
}

@article{shulman:2019:set-theory,
title = {Comparing material and structural set theories},
journal = {Annals of Pure and Applied Logic},
volume = {170},
number = {4},
pages = {465-504},
year = {2019},
issn = {0168-0072},
doi = {https://doi.org/10.1016/j.apal.2018.11.002},
url = {https://www.sciencedirect.com/science/article/pii/S0168007218301295},
author = {Michael Shulman},
}

@article{tait:1967,
	author = {Tait, W. W.},
	publisher = {Association for Symbolic Logic},
	url = {http://www.jstor.org/stable/2271658},
	year = {1967},
	issn = {00224812},
	journal = {Journal of Symbolic Logic},
	number = {2},
	pages = {198--212},
	title = {{Intensional Interpretations of Functionals of Finite Type I}},
	volume = {32},
  doi = {10.2307/2271658}
}

@article{sterling-harper:2021,
  author = {Sterling, Jonathan and Harper, Robert},
  address = {New York, NY, USA},
  publisher ={Association for Computing Machinery},
  year = {2021},
  month = oct,
  doi = {10.1145/3474834},
  eprint = {2010.08599},
  eprintclass = {cs.PL},
  eprinttype = {arXiv},
  issn = {0004-5411},
  journal ={Journal of the ACM},
  number = {6},
  title = {Logical Relations as Types: Proof-Relevant Parametricity for Program Modules},
  volume = {68},
}

@software{sterbac-sterling-artifact,
  author       = {Rapha\"el Sterbac and Jonathan Sterling},
  title        = {Artifact for the paper ``Fuss-free cumulative
                   universes: theory and practice''
                  },
  month        = nov,
  year         = 2026,
  publisher    = {Zenodo},
  doi          = {10.5281/zenodo.24204869},
}
\bibliographystyle{ACM-Reference-Format}

\end{document}